\documentclass{amsart}

\usepackage{amssymb,latexsym,amsfonts,amsmath}
\usepackage{graphicx}
\graphicspath{{Figures/}}
\usepackage{booktabs}
\usepackage{xcolor}
\usepackage{epsfig}
\usepackage{subcaption}
\usepackage{dsfont}
\usepackage{epstopdf}

\newtheorem{theorem}{Theorem}[section]
\newtheorem{lemma}[theorem]{Lemma}
\newtheorem{claim}[theorem]{Claim}

\newtheorem{definition}[theorem]{Definition}

\newtheorem{example}[theorem]{Example}

\newtheorem{remark}[theorem]{Remark}
\numberwithin{equation}{section}

\newcommand{\R}{{\mathbb{R}}}

\newcommand{\N}{{\mathbb N}}

\newcommand{\ie}{{\it i.e.}}

\newcommand{\Ext}{\mathrm{ext}}
\newcommand{\Int}{\mathrm{int}}

\newcommand{\ra}{\rightarrow}

\newcommand{\ol}{\overline}

\newcommand{\dom}{\mathbf{dom}}

\newcommand{\Cont}{\textrm{Cont}}

\begin{document}

\begin{abstract}
In this paper, we focus on mitigating the computational complexity in abstraction-based controller synthesis for interconnected control systems. To do so, we provide a compositional framework for the construction of abstractions for interconnected systems and a bottom-up controller synthesis scheme. In particular, we propose a notion of approximate composition which makes it possible to compute an abstraction of the global interconnected system from the abstractions (possibly of different types) of its components. Finally, by leveraging our notion of approximate composition, we propose a bottom-up approach for the synthesis of controllers enforcing decomposable safety specifications. The effectiveness of the proposed results is demonstrated using two case studies (viz., DC microgrid and traffic network) by comparing them with different abstraction and controller synthesis schemes.           
\end{abstract}

\title[Compositional Abstraction-based Synthesis for Interconnected Systems]{Compositional Abstraction-based Synthesis for Interconnected Systems: An approximate composition approach$^\star$}
\author[A. Saoud]{Adnane Saoud$^{1,\dagger}$}
\author[P. Jagtap]{Pushpak Jagtap$^{2,\dagger}$}
\author[M. Zamani]{Majid Zamani$^{3,4}$}
\author[A. Girard]{Antoine Girard$^1$}
\address{$^1$Laboratoire des Signaux et Systemes (L2S) - CNRS, 91192 Gif sur Yvette, France.}
\email{\{adnane.saoud, antoine.girard\}@l2s.centralesupelec.fr}
\address{$^2$Department of Electrical and Computer Engineering, Technical University of Munich, Germany.}
\email{pushpak.jagtap@tum.de}
\address{$^3$Computer Science Department, University of Colorado Boulder, USA.}
\email{Majid.Zamani@colorado.edu}
\address{$^3$Computer Science Department, Ludwig Maximilian University of Munich, Germany.}
\thanks{$^\dagger$ The authors contributed equally to this work.}
\thanks{$^\star$ This work was supported in part by the H2020 ERC Starting
Grant AutoCPS, the German Research Foundation (DFG) through the grant ZA 873/1-1, and the TUM International Graduate School of Science and Engineering (IGSSE). This project has also received funding from the European Research Council (ERC) under the European Union's Horizon 2020 research and innovation programme (grant agreement No. 725144).}
\maketitle


\section{Introduction}
Control and verification of dynamical systems using discrete abstractions (a.k.a. symbolic models) and formal methods have been an ongoing research area in recent years {(see \cite{tabuada2009verification,belta2017formal} and the references therein)}.
In such approaches, a discrete abstraction (i.e., a system with the finite number of states and inputs) is constructed from the original system. When the concrete and abstract systems are related by some relations such as simulation, alternating simulation or their approximate versions, the discrete controller synthesized for the abstraction can be refined into a hybrid controller for the original system. The use of discrete abstractions principally enables the use of techniques developed in the areas of supervisory control of discrete event systems \cite{cassandras2009introduction} and algorithmic game theory \cite{bloem2012synthesis}. The construction of such discrete abstractions is often based on a discretization of the state and input sets. Due to those sets discretization, symbolic control techniques suffer severely from the curse of dimensionality (i.e, the computational complexity for synthesizing abstractions and controllers grows exponentially with the state and input spaces dimension).

To tackle this problem, several compositional approaches were recently proposed. The authors in \cite{tazaki2008bisimilar} proposed a compositional approach for finite-state abstractions of a network of control systems based on the notion of interconnection-compatible approximate bisimulation. The results in \cite{pola7} provide compositional construction of approximately bisimilar finite abstractions for networks of discrete-time control systems under some incremental stability property. In \cite{mallik2018compositional}, the notion of (approximate) disturbance simulation was used for the compositional synthesis of continuous-time systems, where the states of the neighboring components were modeled as disturbance signals. 
In \cite{zamani2017compositional}, authors provide compositional abstraction using dissipativity approach. 
The authors in~\cite{dallal2015compositional,kim2017small,saoud2018contract,saoud2018contractv,saoud2019contract} use contract-based design and assume-guarantee reasoning to provide compositional construction of symbolic controllers. 

In this paper, we provide a compositional abstraction-based controller synthesis framework for a composition of $N$ control systems. The main contributions of the work are divided into three parts. First, we introduce a notion of approximate composition, while the classical exact composition of components requires the inputs and outputs of neighboring components to be equal, we propose a notion of approximate composition allowing the distance between inputs and outputs of neighboring components to be bounded by a given parameter. The proposed notion enables the composition of control systems (possibly of different types) which allows for more flexibility in the design of the overall symbolic model because each component may be suitable for a particular type of abstraction. Second, with the help of the aforementioned notion, we provide results on the compositional construction of abstractions for interconnected systems. Indeed, given a collection of components, where each concrete component is related to its abstraction by an approximate (alternating) simulation relation, we show how the parameter of the composition of the abstractions needs to be chosen in order to ensure an approximate (alternating) simulation relation between the interconnection of concrete components and the interconnection of discrete ones. Third, we propose a bottom-up approach for symbolic controller synthesis for the composition of $N$ control systems {and decomposable safety specifications} . Finally, we demonstrate the applicability and effectiveness of the results using two case studies (viz., DC microgrid and traffic network) and compare them with different abstraction and controller synthesis schemes in the literature.


\textbf{Related works:}
First attempts to compute compositional abstractions has been proposed for exact simulation relation~\cite{frehse2005compositional,kerber2010compositional} and simulation maps~\cite{tabuada2004compositional}, for which the construction of abstraction exists for restricted class of systems. 
In~\cite{tazaki2008bisimilar}, the first approach to provide compositionality result for approximate relationships was proposed using the notion of interconnection-compatible approximate bisimulation. Different approaches have then been proposed recently using  small-gain (or relaxed small-gain) like conditions~\cite{rungger2018compositional,pola7,swikir2018compositional} and dissipativity property~\cite{zamani2017compositional}. In~\cite{hussien2017abstracting}, a compositional construction of symbolic abstractions was proposed for the class of partially feedback linearizable systems, where the proposed approach relies on the use of a particular type of abstractions proposed in \cite{zamani2012symbolic}.The authors in~\cite{kim2018constructing} present a compositional abstraction procedure for a discrete-time control system by abstracting the interconnection map between different components. {The authors in~\cite{lal2019compositional} propose a compositional approach on the construction of bounded error over-approximations of continuous-time interconnected systems. However, those results are only available for cyclic (cascade and parallel) composition of two components. In~\cite{lal2017safety} the authors propose a compositional bounded error framework for two interconnected hybrid systems. However, they are based on a notion of language inclusion, which differs from the results presented in this paper.}

In parallel, other different approaches have been proposed for compositional controller synthesis. In~\cite{dallal2015compositional}, the authors propose a compositional approach to deal with persistency specifications using Lyapunov-like functions. The authors in~\cite{le2016distributed} use reachability analysis to provide a compositional controller synthesis for discrete-time switched systems and persistency specifications. In~\cite{meyer2017,meyer2018compositional,mallik2018compositional,smith2016interdependence,ghasemi2020compositional} different approaches were proposed for compositional controller synthesis for safety, and more general LTL specifications. All these approaches are based on assume-guarantee reasoning~\cite{8550622} and generally suffer from the underlying conservatism. {In~\cite{8115304}, the authors propose a decentralized approach to the control of networks of incrementally stable systems, enforcing specifications expressed in terms of regular languages, while showing the completeness of their decentralized approach with respect to the centralized one.}

In comparison with existing approaches in the literature, our framework presents the following advantages: 
\begin{itemize}
	\item It allows the use of different types of abstractions for individual components such as continuous~\cite{girard2009hierarchical} or discrete-time abstractions based on state-space quantization \cite{tabuada2009verification}, partition \cite{meyer2015adhs}, covering \cite{5770194}, or without any state-space discretization \cite{zamani2017towards}. {Indeed, given a composed system made of interconnected components, we can deal with the following scenarios:
		\begin{itemize}
			\item The original composed system is continuous-time and the local abstractions are continuous-time. 
			\item The original composed system is discrete-time and the local abstractions are discrete-time.
			\item The original composed system is continuous-time and the local abstractions are discrete-time. In this case, we need to have a uniform time-discretization parameter for all components, which is a common assumption in all the results in the symbolic control literature, except of~\cite{saoud2019contract} where the authors are using continuous-time assume-guarantee contracts to deal with components with different sampling periods.
	\end{itemize}}
	\item We do not need any particular structure of the components such as incremental stability or monotonicity. Moreover, we do not rely on the use of small-gain or dissipativity like conditions;
	\item The proposed approach allows us to develop a bottom-up procedure for controller synthesis which helps to reduce the computational complexity while ensuring completeness with respect to the monolithic synthesis.
\end{itemize}
A preliminary version of this work has been presented in the conference~\cite{SAOUD201813}. The current paper extends the approach in three directions: First, the approach is generalized from cascade interconnections to any composition structure. Second, while in \cite{SAOUD201813} we showed how to build a safety controller using a bottom-up approach, in this paper we show also the completeness of the proposed controller with respect to the monolithic one. Third, while in~\cite{SAOUD201813}, we only presented a simple numerical example, here the theoretical framework is applied to more realistic case studies: DC microgids and road traffic networks.

\section{Transition Systems and Behavioral Relations}
\label{sec:transition}
\textbf{Notations:} The symbols $ \mathbb{N} $, $ \mathbb{N}_0 $, and $\R_0^+ $ denote the set of positive integers, non-negative integers, and non-negative real numbers, respectively. Given sets $X$ and $Y$, we denote by $f:X\ra Y$ an ordinary map from $X$ to $Y$, whereas $f:X\rightrightarrows Y$ denotes set valued map.
For any $x_1,x_2,x_3 \in X$, the map $\mathbf{d}_X:X\times X\rightarrow\R^+_0$ is a pseudometric if the following conditions hold: (i) $x_1=x_2$ implies $\mathbf{d}_X(x_1,x_2)=0$; (ii) $\mathbf{d}_X(x_1,x_2)=\mathbf{d}_X(x_2,x_1)$; (iii) $\mathbf{d}_X(x_1,x_3)\leq \mathbf{d}_X(x_1,x_2)+ \mathbf{d}_X(x_2,x_3)$.
We identify a relation $ \mathcal{R}\subseteq A\times B $ with the map $ \mathcal{R}:A\rightrightarrows B $ defined by $ b\in\mathcal{R}(a) $ if and only if $ (a,b)\in\mathcal{R} $. 
We use notation $\|\cdot\|$ to denote the infinity norm. The null vector of dimension $N \in \mathbb{N}_0$ is denoted by $\mathbf{0}_N:=(0,\ldots,0)$.


First, we introduce the notion of \textit{transition systems} similar to the one provided in \cite{tabuada2009verification}.
\begin{definition}\label{definition6}
	A transition system is a tuple $ S=(X,X^0,U^{\Ext},U^{\Int}, \Delta, Y,H) $, where $ X $ is the set of states $($possibly infinite$)$, $ X^0\subseteq X $ is the set of initial states, $ U^{\Ext} $ is the set of external inputs $($possibly infinite$)$, $ U^{\Int} $ is the set of internal inputs $($possibly infinite$)$, $ \Delta \subseteq X\times U^{\Ext} \times U^{\Int}\times X $ is the transition relation, $ Y $ is the set of outputs, and $ H:X\rightarrow Y $ is the output map. 
\end{definition}
We denote $ x'\in \Delta(x,u^{\Ext},u^{\Int}) $ as an alternative representation for a transition $ (x,u^{\Ext},u^{\Int},x')\in\Delta  $, where state $ x' $ is called a $ (u^{\Ext},u^{\Int}) $-successor (or simply successor) of state $ x $, for some input $ (u^{\Ext},u^{\Int})\in U^{\Ext}\times U^{\Int} $. Given $x\in X$, the set of enabled (admissible) inputs for $x$ is denoted by $U_S^a(x)$ and defined as $U^a_S(x)=\{(u^{\Ext},u^{\Int})\in U^{\Ext}\times U^{\Int} \mid \Delta(x,u^{\Ext},u^{\Int})\neq \emptyset\}$. 
The transition system is said to be:
\begin{itemize}
	\item \textit{pseudometric}, if the input sets $U^i$, $i\in \{\Ext,\Int\}$ and the output set $ Y $ are equipped with pseudometrics $ \mathbf{d}_{U^{i}}: U^i\times U^i \rightarrow\R^+_0 $ and $ \mathbf{d}_Y: Y\times Y \rightarrow\R^+_0 $, respectively.
	\item \textit{finite} (or \textit{symbolic}), if sets $ X $, $ U^{\Int}$, and $ U^{\Ext}$ are finite.
	\item \textit{deterministic}, if there exists at most a $ (u^{\Ext},u^{\Int}) $-successor of $ x $, for any $ x\in X $ and $ (u^{\Ext},u^{\Int})\in U^{\Ext} \times U^{\Int} $.  
\end{itemize}
In the sequel, we consider the approximate relationship for transition systems based on the notion of approximate (alternating) simulation relation to relate abstractions to concrete systems. We start by introducing the notion of approximate simulation relation adapted from~\cite{julius2009approximate}.
\begin{definition}
	\label{Def:feedback}
	Let $S_1=(X_1,X^0_1,U^{\Ext}_1,U^{\Int}_1,\Delta_1,Y_1,H_1)$ and $S_2=(X_2,X^0_2,U^{\Ext}_2,U^{\Int}_2,\Delta_2,Y_2,H_2)$ be two transition systems such that $Y_1$ and $Y_2$ are subsets of the same pseudometric space $Y$ equipped with a pseudometric $\mathbf{d}_Y$ and $U^{\Ext}_j$ $($respectively $U^{\Int}_j$$)$, $j\in \{1,2\}$, are subsets of the same pseudometric space $U^{\Ext}$ $($respectively $U^{\Int})$ equipped with a pseudometric $\mathbf{d}_{U^{\Ext}}$ $($respectively $\mathbf{d}_{U^{\Int}})$. Let $\varepsilon,\mu\geq 0$. A relation $\mathcal R \subseteq X_1\times X_2$ is said to be an $(\varepsilon,\mu)$-approximate simulation relation from $S_1$ to $S_2$, if the following hold:
	\begin{itemize}
		\item[(i)] $\forall x_1^0\in X_1^0$, $\exists  x_2^0\in X_2^0$ such that $(x_1^0,x_2^0)\in \mathcal{R}$;
		\item[(ii)]   $\forall (x_1,x_2)\in \mathcal R$, $\mathbf{d}_Y(H_1(x_1),H_2(x_2))\leq \varepsilon$;
		\item[(iii)]  $\forall (x_1,x_2)\in \mathcal R$, $\forall (u_1^{\Ext},u_1^{\Int})\in {U^a_{S_1}}(x)$, $ \forall x_1'\in  \Delta_1(x_1,u_1^{\Ext},u_1^{\Int})$,  $\exists (u_2^{\Ext},u_2^{\Int})\in {U^a_{S_2}}(x_2) $ with\\ $\max(\mathbf{d}_{U^{\Ext}}(u_1^{\Ext},u_2^{\Ext}),\mathbf{d}_{U^{\Int}}(u_1^{\Int},u_2^{\Int}))\leq \mu$ and $\exists x_2'\in \Delta_2(x_2,u_2^{\Ext},u_2^{\Int})$ satisfying $(x_1',x_2')\in \mathcal R$.
	\end{itemize}
\end{definition}
We denote the existence of an $(\varepsilon,\mu)$-approximate simulation relation from $S_1$ to $S_2$ by $ S_1\preccurlyeq^{\varepsilon,\mu} S_2$.\\

Note that when $\mu=0$ {and $\mathbf{d}_{U^{\Int}}$ is a metric}, we recover the notion of approximate simulation relation introduced in~\cite{girard2007} and when $\mu=\infty$, it is similar to approximate simulation relation given in~\cite{tabuada2009verification}.\\
Approximate simulation relations are generally used for verification problems. If the objective is to synthesize controllers, the notion of approximate alternating simulation relation introduced in \cite{tabuada2009verification} is suitable. Interestingly, the notions of approximate simulation and approximate alternating simulation coincide in the case of deterministic transition systems.

\begin{definition}
	\label{Def:altsimu}
	Let $S_1=(X_1,X^0_1,U^{\Ext}_1,U^{\Int}_1,\Delta_1,Y_1,H_1)$ and $S_2=(X_2,X^0_2,U^{\Ext}_2,U^{\Int}_2,\Delta_2,Y_2,H_2)$ be two transition systems such that $Y_1$ and $Y_2$ are subsets of the same pseudometric space $Y$ equipped with a pseudometric $\mathbf{d}_Y$ and $U^{\Ext}_j$ $($respectively $U^{\Int}_j)$, $j\in \{1,2\}$, are subsets of the same pseudometric space $U^{\Ext}$ $($respectively $U^{\Int})$ equipped with a pseudometric $\mathbf{d}_{U^{\Ext}}$ $($respectively $\mathbf{d}_{U^{\Int}})$. Let $\varepsilon,\mu\geq 0$. A relation $\mathcal R \subseteq X_1\times X_2$ is said to be an $(\varepsilon,\mu)$-approximate alternating simulation relation from $S_2$ to $S_1$, if it satisfies:
	\begin{itemize}
		\item[(i)] $\forall x_2^0\in X_2^0$, $\exists  x_1^0\in X_1^0$ such that $(x_1^0,x_2^0)\in \mathcal{R}$;
		\item[(ii)]   $\forall (x_1,x_2)\in \mathcal R$, $\mathbf{d}_Y(H_1(x_1),H_2(x_2))\leq \varepsilon$;
		\item[(iii)]  $\forall (x_1,x_2)\in \mathcal R$, $\forall (u_2^{\Ext},u_2^{\Int})\in {U^a_{S_2}}(x_2)$,
		$\exists (u_1^{\Ext},u_1^{\Int})\in {U^a_{S_1}}(x_1) $ with $\max(\mathbf{d}_{U^{\Ext}}(u_1^{\Ext},u_2^{\Ext}),\mathbf{d}_{U^{\Int}}(u_1^{\Int},u_2^{\Int}))\leq \mu$ such that $ \forall x_1'\in  \Delta_1(x_1,u_1^{\Ext},u_1^{\Int})$, $\exists x_2'\in  \Delta_2(x_2,u_2^{\Ext},u_2^{\Int})$ satisfying $(x_1',x_2')\in \mathcal R$.
	\end{itemize}
\end{definition}
We denote the existence of an $(\varepsilon,\mu)$-approximate alternating simulation relation from $S_2$ to $S_1$ by $ S_2\preccurlyeq^{\varepsilon,\mu}_{\mathcal{AS}}S_1$.

One can readily see that when $\mu=\infty$ we recover the classical notion of approximate alternating simulation relation as introduced in~\cite{tabuada2009verification}, and when $\mu=0$ {and $\mathbf{d}_{U^{\Int}}$ is a metric}, we recover the notion of approximate simulation relation introduced in~\cite{girard2007} the approximate alternating simulation relation coincides with strong alternating simulation relation given in~\cite{borri2018design}.


\begin{remark}
	Note that the definitions of approximate $($alternating$)$ simulation relations used in this paper are slightly different from the classical ones. Unlike classical definitions, the choice of inputs in our definitions is constrained by some distance property. However, these constraints over inputs are not restrictive and the proposed notions of $(\varepsilon,\mu)$-approximate $($alternating$)$ simulation relations are compatible for different abstraction techniques presented in the literature.
\end{remark}
\section{Networks of Transition Systems and Approximate Composition}\label{sec:approx}

Given a system made of interconnected components, the computation of a direct abstraction of the whole system is computationally expensive. For this reason, we rely here on the notion of approximate composition allowing us to construct a global
abstraction of the system from the abstractions of its components. To analyze the necessity of approximate composition, let us start with the simplest interconnection structure, a cascade composition of two concrete components, where the output of the first system is an input to the second one. When going from concrete (infinite) to abstract (finite) components, the output of the first system and the input to the second system do not coincide any more, since abstractions generally involve quantization of variables. 
To mitigate this mismatch, we introduce a notion of approximate composition, by relaxing the notion of the exact composition and allowing the distance between the output to the first component and input to the second one to be bounded by some given precision.

A network of systems consists of a collection of $N \in \N$ systems $\{S_1,\ldots,S_N\}$, a set of vertices $I=\{1,\ldots,N\}$ and a binary connectivity relation $\mathcal{I}\subseteq I\times I$ where each vertex $i \in I$ is labeled with the system $S_i$. For $i\in I$, we define $\mathcal{N}(i)=\{j\in I\mid (j,i)\in \mathcal{I}\}$ as the set of neighbouring components from which the incoming edges originate. An illustration of a network of interconnected components is given in Figure~\ref{Fig:int}.

\begin{figure*}[t!]
	\centering
	\includegraphics[width=0.6\textwidth]{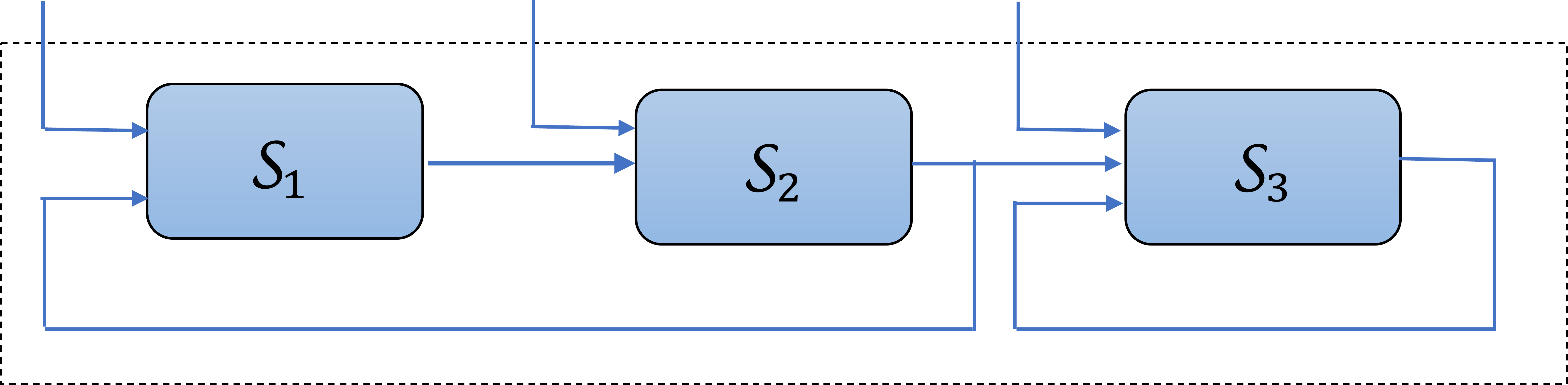}
	\caption{A network of 3 components with $I=\{1,2,3\}$ and a connectivity relation $\mathcal{I}=\{(2,1),(1,2),(2,3),(3,3)\}$.}
	\label{Fig:int}	
\end{figure*}

\begin{definition}
	\label{def:composition}
	Given a collection of transition systems $\{S_i\}_{i\in I}$, where $S_i=(X_i,X^0_i,U^{\Ext}_i,U^{\Int}_i,\Delta_i,Y_i,H_i)$ such that for all $i\in I$, $\prod_{j\in \mathcal{N}(i)}Y_j$ and $U_i^{\Int}$ are subsets of the same pseudometric space equipped with the following pseudometric:
	\begin{align}
	\text{for}~ u^{l,\Int}\hspace{-.2em}=\hspace{-.2em}(y^l_{j_1},\ldots,y^l_{j_k})\hspace{-.1em}, \hspace{-.1em} l\in\{1,2\}\hspace{-.2em},\hspace{-.2em}\text{ with } \mathcal{N}(i)\hspace{-.2em}=\hspace{-.2em}\{j_1,\ldots,\hspace{-.1em}j_k\hspace{-.2em}\},\mathbf d_{U_i^{\Int}}(u^{1,\Int},u^{2,\Int})=\max\limits_{j \in \mathcal{N}(i)}\{\mathbf d_{Y_j}(y^1_j,y^2_j)\} .
	\end{align}
	Let $M:=(\mu_1,\ldots,\mu_N)\in (\R_0^+)^N$. We say that $\{S_i\}_{i\in I}$ is compatible for $M$-approximate composition with respect to $\mathcal{I}$, if for each $i\in I$ and for each $\prod_{j\in \mathcal{N}(i)}\{y_j\} \in \prod_{j\in \mathcal{N}(i)}Y^j$, where the term $\prod_{j\in \mathcal{N}(i)}\{y_j\}$ can be formally defined as $\prod_{j\in \mathcal{N}(i)}\{y_j\}=(y_{j_1},y_{j_2},\ldots,y_{j_p})$ with $\mathcal{N}(i)=\{j_1,j_2,\ldots,j_p\}$, there exists $u^{\Int}_i\in U^{\Int}_i$ such that $\mathbf d_{U^{\Int}_i}(u^{\Int}_i,\prod_{j\in \mathcal{N}(i)}\{y_j\}) \leq \mu_i$.
	We denote $M$-approximate composed system by $\langle S_i \rangle_{i\in I}^{M,\mathcal{I}}$ and is given by the tuple $\langle S_i \rangle_{i\in I}^{M,\mathcal{I}}=(X,X^0,U^{\Ext},\Delta_M,Y,H)$, where:
	\begin{itemize}
		\item $X=\prod_{i \in I}X_i$; $X^0=\prod_{i \in I}X_i^0$;
		\item $U^{\Ext}=\prod_{i \in I}U^{\Ext}_i$; $Y=\prod_{i \in I}Y_i$;
		\item $H(x)=H(x_1,\ldots,x_N)=(H_1(x_1),\ldots,H_N(x_N))$;
		\item for $x=(x_1,\ldots,x_N)$, $x'=(x_1',\ldots,x_N')$ and $u^{\Ext}=(u^{\Ext}_1,\ldots,u^{\Ext}_N)$, $x' \in \Delta_M(x,u^{\Ext})$ if and only if for all $i\in I$, and for all $\prod_{j\in \mathcal{N}(i)}\{y_j\}= \prod_{j\in \mathcal{N}(i)}\{H_j(x_j)\}  \in \prod_{j\in \mathcal{N}(i)}Y_j$, there exists $u^{\Int}_i\in U^{\Int}_i$ with $\mathbf d_{U^{\Int}_i}(u^{\Int}_i,\prod_{j\in \mathcal{N}(i)}\{y_j\}) \leq \mu_i$, $(u_i^{\Ext},u_i^{\Int}) \in U_{S_i}^a(x_i)$ and  $x_i'\in \Delta_i(x_i,u^{\Ext}_i,u^{\Int}_i)$.
	\end{itemize}
\end{definition}

For the sake of simplicity of notations, we use $S_M$ instead of $\langle S_i \rangle_{i\in I}^{M,\mathcal{I}}$ throughout the paper. {Note that since all the internal inputs of a component are outputs of other components, the set of internal inputs of the composed system is simply the empty set. Hence, to improve readability, the composed system is defined without an internal input set. However, the same approach could be readily generalized to deal with a composed system with a given internal input set.} If $M=\mathbf{0}_N$ and $\mathbf{d}_{U_i^{\Int}}$ is a metric, for all $i\in I$, we say that collection of systems $\{S_i\}_{i\in I}$ is compatible for exact composition. Let us remark that, for the composed system, the set of enabled inputs will be defined with respect to the set $U^{\Ext}$. 
We equip the composed output space with the metric:
\begin{align}
\label{eqn:metric_output}
\text{for}\;\; y^j\in Y \text{ with } y^j=(y_1^j,\ldots,y_N^j), j\in \{1,2\}, \mathbf d_Y(y^1,y^2)=\max\limits_{i \in I}\{\mathbf d_{Y_i}(y^1_i,y^2_i)\}.
\end{align}
Similarly, we equip the composed input space with metric: 
\begin{align}\label{eqn:metric_input}
\text{for}\;\; u^j\in U^{\Ext} \text{ with } u^j=(u_1^j,\ldots,u_N^j), j\in \{1,2\},
\mathbf d_{U^{\Ext}}(u^1,u^2)=\max\limits_{i \in I}\{\mathbf d_{U_i^{\Ext}}(u^1_i,u^2_i)\}.
\end{align}

Let us remark that the parameter of the composition, i.e. $M$, affects the conservativeness of the composed transition system. The following result shows that by increasing the parameter of the composition, the composed transition system allows for more nondeterminism in transitions and hence becomes more conservative. This result is straightforward and is stated without any proof.

\begin{claim}
	\label{claim:1}
	Consider a collection of systems $\{S_i\}_{i\in I}$ and $M=(\mu_1,\ldots,\mu_N)\in (\R^+_0)^N$. If $\{S_i\}_{i\in I}$ is compatible for $M$-approximate composition with respect to $\mathcal{I}$, then it is also compatible for $\overline{M}$-approximate composition with respect to $\mathcal{I}$, for any $\ol M=(\ol\mu_1,\ldots,\ol\mu_N)\in (\R^+_0)^N$ such that $\overline{M}\geq M$ $($i.e., $\ol\mu_i\geq\mu_i$, $i\in I)$. Moreover, the relation $\mathcal{R}=\{(x,x')\in  X\times X \mid x=x'\}$ is a $(0,0)$-approximate simulation relation from $S_M$ to $S_{\overline{M}}$, where $S_M=\langle S^i \rangle_{i\in I}^{M,\mathcal{I}}$ and $S_{\overline{M}}=\langle S_i \rangle_{i\in I}^{\overline{M},\mathcal{I}}$.
\end{claim}

\section{Compositionality Results}
\label{sec:compo}
In this section, we provide relations between interconnected systems based on the relations between their components. We first present the compositionality result for approximate alternating simulation relations.

\begin{theorem}
	\label{thm:altsimulation}
	Let $\{S_i\}_{i\in I}$ and $\{\hat{S}_i\}_{i\in I}$ be two collection of transition systems with $S_i=(X_i,X^0_i,U^{\Ext}_i,U^{\Int}_i,\Delta_i$, $Y_i,H_i)$ and $\hat{S}_i=(\hat{X}_i,\hat{X}^0_i,\hat{U}^{\Ext}_i,\hat{U}^{\Int}_i,\hat{\Delta}_i,\hat{Y}_i,\hat{H}_i)$. Consider non-negative constants $\varepsilon_i,\mu_i,\delta_i \geq0$, $i\in I$, with $\varepsilon=\max\limits_{i\in I}\varepsilon_i$ and $\mu=\max\limits_{i\in I}\mu_i$. Let the following conditions hold:
	\begin{itemize}
		\item for all $i\in I$, $\hat{S}_i\preccurlyeq^{\varepsilon_i,\mu_i}_{\mathcal{AS}}S_i$ with a relation $\mathcal{R}_i$;
		\item $\{S_i\}_{i\in I}$ are compatible for $M$-approximate composition with respect to $\mathcal{I}$, with $M=(\delta_1,\ldots,\delta_N)$;
		\item $\{\hat{S}_i\}_{i\in I}$ are compatible for $\hat{M}$-approximate composition with respect to $\mathcal{I}$, with $\hat{M}\hspace{-.2em}=\hspace{-.2em}(\mu_1\hspace{-.1em}+\hspace{-.1em}\delta_1\hspace{-.1em}+\hspace{-.1em}\varepsilon,\ldots,\mu_N\hspace{-.1em}+\hspace{-.1em}\delta_N\hspace{-.1em}+\hspace{-.1em}\varepsilon)$.
	\end{itemize}
	Then the relation $\mathcal{R} \subseteq X\times  \hat{X}$ defined by
	\begin{align*}
	\mathcal{R}\hspace{-.2em}=\hspace{-.2em}\{(x_1,\ldots,x_N,\hat{x}_1,\ldots,\hat{x}_N)\hspace{-.2em}\in\hspace{-.2em}  X\hspace{-.2em}\times\hspace{-.2em} \hat{X} \mid \forall i \hspace{-.2em}\in\hspace{-.2em} I, (x_i,\hat{x}_i)\hspace{-.2em}\in\hspace{-.2em} \mathcal{R}_i\}
	\end{align*}
	is an $(\varepsilon, \mu)$-approximate alternating simulation relation from $\hat{S}_{\hat{M}}$ to $S_M$ $($\ie, $\hat{S}_{\hat{M}}\preccurlyeq^{\varepsilon,\mu}_{\mathcal{AS}}S_M)$, where $S_M=\langle S_i \rangle_{i\in I}^{M,\mathcal{I}}$ and $\hat{S}_{\hat{M}}=\langle \hat{S}_i \rangle_{i\in I}^{\hat{M},\mathcal{I}}$.
\end{theorem}
\begin{proof}
	The first condition of Definition \ref{Def:altsimu} is directly satisfied.
	Let $(x,\hat{x})\in\mathcal{R}$ with $x=(x_1,\ldots,x_N)$ and $\hat{x}=(\hat{x}_1,\ldots,\hat{x}_N)$. We have 
	$\mathbf{d}_{Y}(H(q),\hat{H}(\hat{q}))= \mathbf{d}_{Y}((H_1(q_1),\ldots,H_N(q_N)),(\hat{H}_1(\hat{q}_1),\ldots,\hat{H}_N(\hat{q}_N)))=\max\limits_{i\in I}  \mathbf d_{Y_i}(H_i(q_i),\hat{H}(\hat{q}_i)) \leq \max\limits_{i\in I}\varepsilon_i= \varepsilon$,
	where the first equality comes from the definition of the output map for approximate composition, the second equality follows from (\ref{eqn:metric_output}), and the inequality comes from condition (ii) of Definition \ref{Def:altsimu}.
	
	Consider $(x,\hat{x})\in\mathcal{R}$ with $x=(x_1,\ldots,x_N)$ and $\hat{x}=(\hat{x}_1,\ldots,\hat{x}_N)$ and any $\hat{u}^{\Ext} \in \hat{U}^a_{\hat{S}_{\hat{M}}}(x)$ with $\hat{u}^{\Ext}=(\hat{u}_1^{\Ext},\ldots,\hat{u}_N^{\Ext})$. Let us prove the existence of $u^{\Ext} \in U^a_{S_M}(x)$ with $\mathbf{d}_{U^{\Ext}}(u^{\Ext},\hat{u}^{\Ext}) \leq \mu$ and such that for any $x'\in \Delta_M(x,u^{\Ext})$, there exists $\hat{x}' \in \hat{\Delta}_{\hat M}(\hat{x},\hat{u})$ satisfying $(x',\hat{x}')\in \mathcal{R}$.
	From the definition of relation $\mathcal R$, we have for all $i\in I$, $(x_i,\hat{x}_i)\in \mathcal{R}_i$, then from the third condition of Definition \ref{Def:altsimu}, we have for all $(\hat{u}_i^{\Ext},\hat{u}_i^{\Int}) \in \hat{U}^a_{\hat{S}_i}(\hat{x}_i)$, the existence of $(u_i^{\Ext},u_i^{\Int})\in U^a_{S_i}(x_i)$ with $\mathbf{d}_{U_i^{\Ext}}(u_i^{\Ext},\hat{u}_i^{\Ext})\leq\mu_i$ and $\mathbf{d}_{U_i^{\Int}}(u_i^{\Int},\hat{u}_i^{\Int})\leq\mu_i$ such that for any $x_i'\in \Delta_i(x_i,u_i^{\Ext},u_i^{\Int})$ there exists $\hat{x}_i'\in \hat{\Delta}_i(\hat{x}_i,\hat{u}_i^{\Ext},\hat{u}_i^{\Int})$ such that $(x_i',\hat{x}_i')\in \mathcal{R}_i$.
	
	Let us show that the input $\hat{u}^{\Int}=(\hat{u}^{\Int}_1,\ldots,\hat{u}^{\Int}_N)$ satisfies the requirement of the $\hat{M}$-approximate composition of the components $\{\hat{S}_i\}_{i\in I}$. The condition $\mathbf{d}_{U_i^{\Int}}(u_i^{\Int},\hat{u}_i^{\Int})\leq \mu_i$ implies that
	\begin{equation*}
	\begin{split}
	\mathbf{d}_{U_i^{\Int}}\hspace{-.1em}(\hat{u}_i^{\Int},\hspace{-.6em}\prod_{j\in \mathcal{N}(i)}\hspace{-.6em}\{\hat{y}_j\}) &\hspace{-.2em}\leq\hspace{-.2em} \mathbf{d}_{U_i^{\Int}}(\hat{u}_i^{\Int},u_i^{\Int})\hspace{-.2em}+\hspace{-.2em}\mathbf{d}_{U_i^{\Int}}(u_i^{\Int},\hspace{-.6em}\prod_{j\in \mathcal{N}(i)}\hspace{-.6em}\{\hat{y}_j\})\\& \hspace{-.2em}\leq\hspace{-.2em} \mathbf{d}_{U_i^{\Int}}\hspace{-.1em}(\hat{u}_i^{\Int}\hspace{-.2em},\hspace{-.1em}u_i^{\Int}\hspace{-.1em})\hspace{-.2em}+\hspace{-.2em}\mathbf{d}_{U_i^{\Int}}\hspace{-.1em}(\hspace{-.1em}u_i^{\Int}\hspace{-.2em},\hspace{-.7em}\prod_{j\in \mathcal{N}(i)}\hspace{-.7em}\{y_j\}\hspace{-.2em})\hspace{-.2em}+\hspace{-.2em}\mathbf{d}_{U_i^{\Int}}\hspace{-.1em}(\hspace{-.6em}\prod_{j\in \mathcal{N}(i)}\hspace{-.7em}\{y_j\}\hspace{-.1em},\hspace{-.6em}\prod_{j\in \mathcal{N}(i)}\hspace{-.7em}\{\hat{y}_j\}\hspace{-.2em})\\& \hspace{-.2em}\leq \mu_i + \delta_i+ \max\limits_{j\in \mathcal{N}(i)} \varepsilon_j \leq \mu_i + \delta_i+ \max\limits_{j \in I} \varepsilon_j =\mu_i+\delta_i+\varepsilon.
	\end{split}
	\end{equation*}
	Hence, from (iii) the $\hat{M}$- approximate composition with respect to $\mathcal{I}$ of $\{\hat{S}_i\}_{i\in I}$ is well defined in the sense of Definition \ref{def:composition}. Thus, condition (iii) in Definition \ref{Def:altsimu} holds with $u^{\Ext}=(u_1^{\Ext},\ldots,u_N^{\Ext})$ satisfying $\mathbf d_{U^{\Ext}}(u^{\Ext},\hat{u}^{\Ext})=\max\limits_ {i \in I}\{\mathbf d_{U_i^{\Ext}}(u^{\Ext}_i,\hat{u}^{\Ext}_i)\}  = \max\limits_{i \in I}\{\mu_i\} =\mu$, and one obtains $\hat{S}_{\hat{M}}\preccurlyeq^{\varepsilon,\mu}_{\mathcal{AS}}S_M$. 
\end{proof}

We then have the following compositionality result for approximate simulation relations.

\begin{theorem}
	\label{thm:simulation}
	Let $\{S_i\}_{i\in I}$ and $\{\hat{S}_i\}_{i\in I}$ be two collection of transition systems with $S_i=(X_i,X^0_i,U^{\Ext}_i,U^{\Int}_i,\Delta_i$, $Y_i,H_i)$ and $\hat{S}_i=(\hat{X}_i,\hat{X}^0_i,\hat{U}^{\Ext}_i,\hat{U}^{\Int}_i,\hat{\Delta}_i,\hat{Y}_i,\hat{H}_i)$. Consider non-negative constants $\varepsilon_i,\mu_i,\delta_i \geq0$, $i\in I$, with $\varepsilon=\max\limits_{i\in I}\varepsilon_i$ and $\mu=\max\limits_{i\in I}\mu_i$. Let the following conditions hold:
	\begin{itemize}
		\item for all $i\in I$, $S_i\preccurlyeq^{\varepsilon_i,\mu_i}\hat S_i$ with a relation $\mathcal{R}_i$;
		\item $\{S_i\}_{i\in I}$ are compatible for $M$-approximate composition with respect to $\mathcal{I}$, with $M=(\delta_1,\ldots,\delta_N)$;
		\item $\{\hat{S}_i\}_{i\in I}$ are compatible for $\hat{M}$-approximate composition with respect to $\mathcal{I}$, with $\hat{M}\hspace{-.2em}=\hspace{-.2em}(\mu_1\hspace{-.1em}+\hspace{-.1em}\delta_1\hspace{-.1em}+\hspace{-.1em}\varepsilon,\ldots,\mu_N\hspace{-.1em}+\hspace{-.1em}\delta_N\hspace{-.1em}+\hspace{-.1em}\varepsilon)$.
	\end{itemize}
	Then the relation $\mathcal{R} \subseteq X\times  \hat{X}$ defined by
	\begin{align*}
	\mathcal{R}\hspace{-.2em}=\hspace{-.2em}\{(x_1,\ldots,x_N,\hat{x}_1,\ldots,\hat{x}_N)\hspace{-.2em}\in\hspace{-.2em}  X\hspace{-.2em}\times\hspace{-.2em} \hat{X} \mid \hspace{-.2em}\forall i \hspace{-.2em}\in\hspace{-.2em} I, (x_i,\hat{x}_i)\hspace{-.2em}\in \hspace{-.2em}\mathcal{R}_i\}
	\end{align*}
	is an $(\varepsilon, \mu)$-approximate simulation relation from $S_M$ to $\hat{S}_{\hat{M}}$ $($\ie, $S_M\preccurlyeq^{\varepsilon,\mu}\hat S_{\hat M})$, where $S_M=\langle S_i \rangle_{i\in I}^{M,\mathcal{I}}$ and $\hat{S}_{\hat{M}}=\langle \hat{S}_i \rangle_{i\in I}^{\hat{M},\mathcal{I}}$.
\end{theorem}
 \begin{proof}
	The first condition of Definition \ref{Def:feedback} is directly satisfied.
	Let $(x,\hat{x})\in\mathcal{R}$ with $x=(x_1,\ldots,x_N)$ and $\hat{x}=(\hat{x}_1,\ldots,\hat{x}_N)$. We have 
	$\mathbf{d}_{Y}(H(x),\hat{H}(\hat{x}))= \mathbf{d}_{Y}((H_1(x_1),\ldots,H_N(x_N)),(\hat{H}_1(\hat{x}_1),\ldots,\hat{H}_N(\hat{x}_N)))=\max\limits_{i\in I}  \mathbf d_{Y_i}(H_i(x_i),\hat{H}(\hat{x}_i)) \leq \max\limits_{i\in I}\varepsilon_i= \varepsilon,$
	where the first equality comes from the definition of the output map for approximate composition, the second equality follows from (\ref{eqn:metric_output}), and the inequality comes from the second condition of Definition \ref{Def:feedback}.

	Consider $(x,\hat{x})\in\mathcal{R}$ with $x=(x_1,\ldots,x_N)$ and $\hat{x}=(\hat{x}_1,\ldots,\hat{x}_N)$ and any $u^{\Ext} \in U_{S_M}^a(x)$ with $u^{\Ext}=(u_1^{\Ext},\ldots,u_N^{\Ext})$. Consider the transition $x'\in \Delta_M(x,u^{\Ext})$. This implies that for all $i\in I$, and for all $\prod_{j\in \mathcal{N}(i)}\{y_j\}= \prod_{j\in \mathcal{N}(i)}\{H_j(x_j)\}  \in \prod_{j\in \mathcal{N}(i)}Y_j$, there exists $u^{\Int}_i\in U^{\Int}_i$ with $\mathbf d_{U^{\Int}_i}(u^{\Int}_i,\prod_{j\in \mathcal{N}(i)}\{y_j\}) \leq \delta_i$, $(u_i^{\Ext},u_i^{\Int}) \in U^a_{S_i}(x_i)$ and  $x_i'\in \Delta_i(x_i,u^{\Ext}_i,u^{\Int}_i)$. Let us prove the existence of an input $\hat{u}^{\Ext} \in \hat{U}^a_{\hat{S}_{\hat{M}}}(\hat{x})$ such that $\mathbf{d}_{U^{\Ext}}(u^{\Ext},\hat{u}^{\Ext}) \leq \mu$ and a transition $\hat{x}'\in \hat{\Delta}_{\hat M}(\hat{x},\hat{u}^{\Ext})$ such that $(x',\hat{x}') \in \mathcal{R}$.\\
	From the definition of the relation $\mathcal R$, we have for all $i\in I$, $(x_i,\hat{x}_i)\in \mathcal{R}_i$, $(u_i^{\Ext},u_i^{\Int}) \in U^a_{S_i}(x_i)$ and $x_i'\in \Delta_i(x_i,u_i^{\Ext},u_i^{\Int})$, then from the third condition of the Definition \ref{Def:feedback}, there exists $(\hat{u}_i^{\Ext},\hat{u}_i^{\Int})\in \hat{U}^a_{\hat{S}_i}(\hat{x}_i)$ with $\mathbf{d}_{U_i^{\Ext}}(u_i^{\Ext},\hat{u}_i^{\Ext})\leq\mu_i$ and $\mathbf{d}_{U_i^{\Int}}(u_i^{\Int},\hat{u}_i^{\Int})\leq\mu_i$  and there exists $\hat{x}_i'\in \hat{\Delta}_i(\hat{x}_i,\hat{u}_i^{\Ext},\hat{u}_i^{\Int})$ such that $(x_i',\hat{x}_i')\in \mathcal{R}_i$. 
	Let us show that the input $\hat{u}^{\Int}=(\hat{u}^{\Int}_1,\ldots,\hat{u}^{\Int}_N)$ satisfies the requirement of the $\hat{M}$-approximate composition of the components $\{\hat{S}_i\}_{i\in I}$. The condition $\mathbf{d}_{U_i^{\Int}}(u_i^{\Int},\hat{u}_i^{\Int})\leq \mu_i$ implies that
	\begin{equation*}
	\begin{split}
	&\mathbf{d}_{U_i^{\Int}}\hspace{-.1em}(\hat{u}_i^{\Int},\hspace{-.6em}\prod_{j\in \mathcal{N}(i)}\hspace{-.6em}\{\hat{y}_j\}) \hspace{-.2em}\leq\hspace{-.2em} \mathbf{d}_{U_i^{\Int}}(\hat{u}_i^{\Int},u_i^{\Int})\hspace{-.2em}+\hspace{-.2em}\mathbf{d}_{U_i^{\Int}}(u_i^{\Int},\hspace{-.6em}\prod_{j\in \mathcal{N}(i)}\hspace{-.6em}\{\hat{y}_j\})\\& \hspace{-.2em}\leq\hspace{-.2em} \mathbf{d}_{U_i^{\Int}}\hspace{-.1em}(\hat{u}_i^{\Int}\hspace{-.2em},\hspace{-.1em}u_i^{\Int}\hspace{-.1em})\hspace{-.2em}+\hspace{-.2em}\mathbf{d}_{U_i^{\Int}}\hspace{-.1em}(\hspace{-.1em}u_i^{\Int}\hspace{-.2em},\hspace{-.7em}\prod_{j\in \mathcal{N}(i)}\hspace{-.7em}\{y_j\}\hspace{-.2em})\hspace{-.2em}+\hspace{-.2em}\mathbf{d}_{U_i^{\Int}}\hspace{-.1em}(\hspace{-.6em}\prod_{j\in \mathcal{N}(i)}\hspace{-.7em}\{y_j\}\hspace{-.1em},\hspace{-.6em}\prod_{j\in \mathcal{N}(i)}\hspace{-.7em}\{\hat{y}_j\}\hspace{-.2em})\\& \hspace{-.2em}\leq \mu_i + \delta_i+ \max\limits_{j\in \mathcal{N}(i)} \varepsilon_j \leq \mu_i + \delta_i+ \max\limits_{j \in I} \varepsilon_j =\mu_i+\delta_i+\varepsilon.
	\end{split}
	\end{equation*}
	Hence, the $\hat{M}$- approximate composition with respect to $\mathcal{I}$ of $\{\hat{S}_i\}_{i\in I}$ is well defined in the sense of Definition \ref{def:composition}. Thus, condition (iii) in Definition \ref{Def:feedback} holds with $\hat{u}^{\Ext}=(\hat{u}_1^{\Ext},\ldots,\hat{u}_N^{\Ext})$ satisfying $\mathbf d_{U^{\Ext}}(u^{\Ext},\hat{u}^{\Ext})=\Vert \prod_{i \in I}\{\mathbf d_{U_i^{\Ext}}(u^{\Ext}_i,\hat{u}^{\Ext}_i)\} \Vert = \Vert \prod_{i \in I}\{\mu_i\} \Vert=\mu$ and $\hat{x}'=(\hat{x}_1',\ldots,\hat{x}_N')$, and one obtains $S_M\preccurlyeq^{\varepsilon,\mu}\hat{S}_{\hat{M}}$.  
\end{proof}

Intuitively, the results of the previous theorems can be interpreted as follows: the result in Theorem~\ref{thm:simulation} can be used for compositional verification. Given a collection of systems $\{S_i\}_{i\in I}$, if each component approximately satisfies a specification $Q_i$ ($S_i \preccurlyeq^{\varepsilon_i,\mu_i} Q_i$), then the composed system $S_M=\langle S_i \rangle_{i\in I}^{M,\mathcal{I}}$ approximately satisfies a composed specification $Q=\langle Q_i \rangle_{i\in I}^{\hat{M},\mathcal{I}}$ ($S \preccurlyeq^{\varepsilon,\mu} Q$). Note that for constructing controllers, the results of Theorem \ref{thm:altsimulation} is more suitable. Given a collection of components $\{S_i\}_{i\in I}$, for $i\in I$, let $\hat{S}_i$ an abstraction for $S_i$ ($\hat{S}_i\preccurlyeq^{\varepsilon_i,\mu_i}_{\mathcal{AS}}S_i$), then the composed system $\hat{S}_{\hat M}=\langle \hat{S}_i \rangle_{i\in I}^{\hat{M},\mathcal{I}}$ is an abstraction of the system $S_M=\langle S_i \rangle_{i\in I}^{M,\mathcal{I}}$ ($\hat{S}_{\hat M}\preccurlyeq^{\varepsilon,\mu}_{\mathcal{AS}}S_M$). Figure~\ref{fig:thms} illustrates these results.
\begin{figure}[t!]
	\begin{center}
		\includegraphics[scale=0.9]{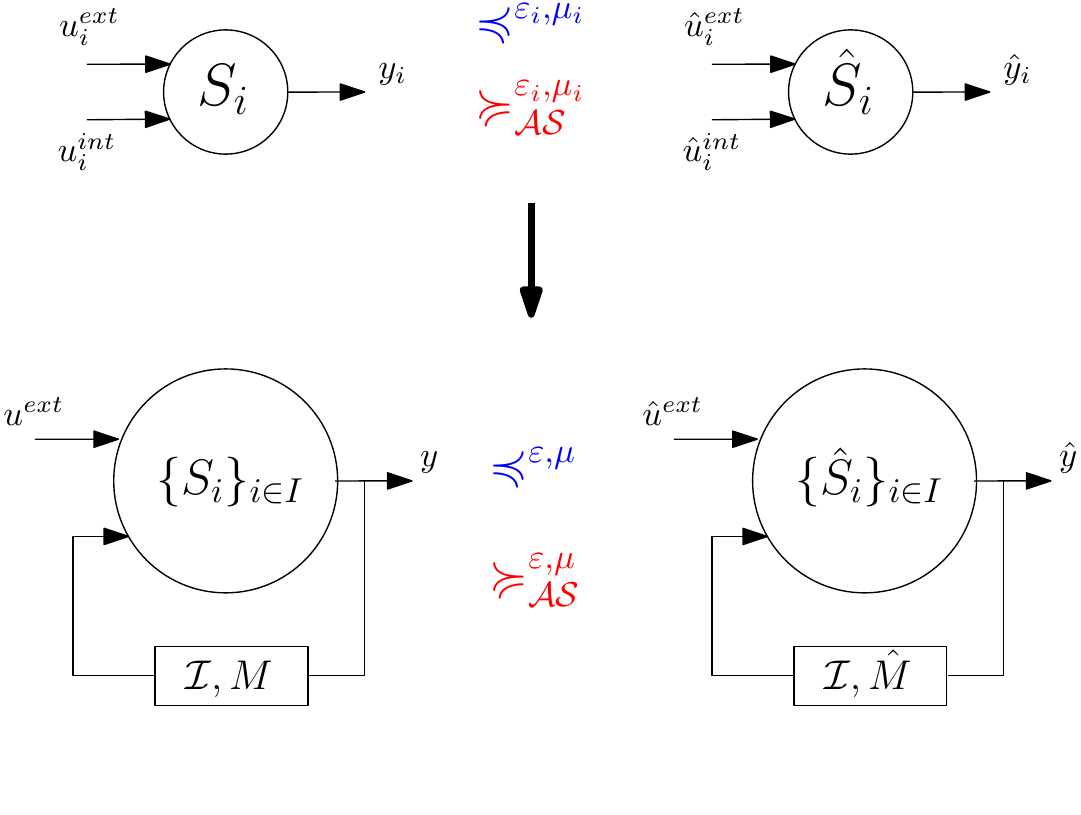}
	\end{center}
	\vspace{-0.5cm}
	\caption{Illustration of compositionality results for a collection of transition systems using the notion of approximate composition and approximate (alternating) simulation relations as formalized in Theorems~\ref{thm:altsimulation} and~\ref{thm:simulation}.}
	\label{fig:thms}
	\vspace{-0.3cm}
\end{figure}
\begin{remark}
	In symbolic control literature, different approaches have been presented to compute $($in$)$finite abstractions for different classes of systems including linear systems~\cite{belta2017formal,girard2009hierarchical}, monotone $($or mixed-monotone$)$ systems~\cite{coogan2015efficient,meyer2015adhs}, time-delay systems~\cite{pola2010symbolic}, switched systems~\cite{girard2009approximately}, incrementally stable $($or stabilizable$)$ systems~\cite{pola2008approximately}, incrementally stable stochastic $($switched$)$ systems~\cite{zamani2017towards} and incrementally stable time-delayed $($stochastic$)$ control systems~\cite{pola2010symbolic,jagtap2020symbolic}.
	Let us point out that the proposed compositional framework in this paper is suitable for different types of $($in$)$finite abstractions which allows for modularity and flexibility in the construction of the symbolic models.
\end{remark}


\section{Bottom-up safety controller synthesis}
\label{sec:control}

{In the previous section, we have shown how to derive approximate (alternating) simulation relations
	between the concrete and abstract models of the whole system from those of the components, which can mainly be used for the construction of compositional abstractions. Based on those results, in this section, we go one step further and show how the proposed compositionality results make it possible to provide a bottom-up approach for the synthesis of controllers enforcing decomposable safety specifications.}

\subsection{Controlled systems}
Consider a system $S=(X,X^0,U^{\Ext},U^{\Int},\Delta,Y,H)$ and a memoryless controller $C:X \rightrightarrows{U^{\Ext}\times U^{\Int}}$ such that for all $x\in X$, $C(x)\subseteq U^a_S(x)$. Let $\dom(C)$ be the domain of controller defined by $\dom(C)=\{x \in X \mid C(x)\neq \emptyset\} \subseteq X$. We define a controlled transition system by a tuple $S|C=(X_C, X_C^0,U_C^{\Ext},U_C^{\Int},\Delta_C,Y_C,H_C)$, where:
\begin{itemize}
	\item $X_C=X\cap \dom(C)$ is the set of states;
	\item $ X_C^0=X^0\cap \dom(C)$ is the set of initial states;
	\item $U_C^{\Ext}=U^{\Ext}$ is the set of external inputs;
	\item $U_C^{\Int}=U^{\Int}$ is the set of internal inputs;
	\item $Y_C=Y$ is the set of outputs;
	\item $H_C=H$ is the output map;
	\item a transition relation: $x_C'\in\Delta_C(x_C,u_C^{\Ext},u_C^{\Int})$ if and only if $x_C'\in\Delta(x_C,u_C^{\Ext},u_C^{\Int})$ and $(u_C^{\Ext},u_C^{\Int})\in C(x_C)$.	
\end{itemize}

We first introduce the following auxiliary lemma relating the system $S$ and the controlled system $S|C$.
\begin{lemma}
	\label{lem:controlledsystem}
	Given the systems $S$ and $S|C$ defined above, we have that $S|C \preccurlyeq^{0,0}_{\mathcal{AS}} S$.
\end{lemma}

\begin{proof}
		Define the relation $\mathcal{R}=\{(x,\overline{x}) \in X \times X_C \mid x=\overline{x}\}$. We have that $ X_C^0=X^0\cap \dom(C)\subseteq X^0$, hence the first condition of Definition~\ref{Def:altsimu} is satisfied. Let $(x,\overline{x}) \in \mathcal{R}$. We have that $\mathbf d_Y(H(x),H_C(\overline{x}))=\mathbf d_Y(H(x),H(x))=0$ which shows condition (ii) of Definition~\ref{Def:altsimu}. Now consider $(x,\overline{x}) \in \mathcal{R}$ and any $(\overline u^{\Ext},\overline u^{\Int}) \in U_{S|C}^a(\overline{x})$. We choose $u^{\Ext}=\overline u^{\Ext}$ and $u^{\Int}=\overline u^{\Int}$ with $(u^{\Ext},u^{\Int}) \in U_{S}^a({x})$ and $(u^{\Ext},u^{\Int}) \in C(\overline{x})$. Then for all ${x}'\in \Delta({x},u^{\Ext},u^{\Int})$ we have $x'\in \Delta_C({x},u^{\Ext},u^{\Int})=\Delta_C(\overline x,u^{\Ext},u^{\Int})$. Since $x'\in \Delta_C(\overline x,u^{\Ext},u^{\Int})$, we have the existence of $\overline x' \in \Delta_C(\overline x,u^{\Ext},u^{\Int})$ satisfying  $\overline x'=x'$. This implies $(x',\overline{x}')  \in \mathcal{R}$, which concludes the proof.  
\end{proof}

\subsection{Safety controller}

\begin{definition}
	\label{safety}
	A safety controller $C$ for the transition system $S$ and the safe set $\mathfrak S\subseteq X$ satisfies:
	\begin{itemize}
		\item[(i)] $\dom(C) \subseteq \mathfrak S$;
		\item[(ii)]  $\forall x \hspace{-0.2em}\in \hspace{-0.2em}\dom(C)$ and $\forall (u^{\Ext},u^{\Int})\hspace{-0.2em}\in\hspace{-0.2em} C(x)$,  $\Delta(x,u^{\Ext},u^{\Int}) \hspace{-0.2em}\subseteq \hspace{-0.2em}\dom(C)$.
	\end{itemize}
\end{definition}

There are in general several controllers that solve the safety problem. A suitable solution is a controller that enables as many actions as possible. Such controller $C^*$ is referred to as \textit{a maximal safety controller}, in the sense that for any other safety controller $C$ and for all $x\in X$, we have $C(x)\subseteq C^*(x)$. In order to define carefully the maximal safety controller, we introduce the concept of a controlled invariant set.

\begin{definition}
	\label{def:inv}
	Consider a transition system $S$ and a safe set $\mathfrak S \subseteq X$. A subset $A \subseteq \mathfrak S$ is said to be a controlled invariant if for all $x \in A$ there exists $(u^{\Ext},u^{\Int}) \in U^{\Ext} \times U^{\Int}$ such that $\Delta(x,u^{\Ext},u^{\Int})\subseteq A$.
\end{definition}

It was shown in~\cite{tabuada2009verification} that there exists a maximal controlled invariant $\Cont(\mathfrak S)$ which is the union of all controlled invariants. The \textit{maximal safety controller} can be defined as follows: 
\begin{itemize}
	\item for all $x\notin \Cont(\mathfrak S)$, $C^*(x)= \emptyset$;
	\item for all $x\in \Cont(\mathfrak S)$, $C^*(x)=\{(u^{\Ext},u^{\Int}) \in U^a(x)\mid \Delta(x,u^{\Ext},u^{\Int})\subseteq \Cont(\mathfrak S) \}$.
\end{itemize}
Let us remark that for any safety controller $C$ we have that $\dom(C) \subseteq \Cont(\mathfrak S)$, while for the \textit{maximal safety controller}  $C^*$, we have  $\dom(C^*) = \Cont(\mathfrak S)$.

\subsection{Bottom-up synthesis of controllers}
The size of transition systems is crucial for computational efficiency of discrete safety controller synthesis algorithms. As the size of transition systems grows, the classical safety synthesis suffers from the curse of dimensionality. In this subsection, we show how to synthesize safety controllers for interconnected systems using a bottom-up approach. Consider a global system $S$ made of $N$ interconnected components $S_i$, $i \in I$, and a global decomposable safety specification $\mathfrak S=\mathfrak S_1\times \ldots\times \mathfrak S_N$. We start by synthesizing a local safety controller $C_i$ for each component $S_i$ and safety specification $\mathfrak S_i$, compose the local controlled components (by computing $\langle S_i|C_i \rangle_{i\in I}^{M,\mathcal{I}}$), and then synthesize a global safety controller for $\langle S_i|C_i \rangle_{i\in I}^{M,\mathcal{I}}$ against the safety specification $\mathfrak S$. We first give an example illustrating the idea of bottom-up safety synthesis.

\begin{figure}[h]
	\begin{center}
		\includegraphics[scale=0.12]{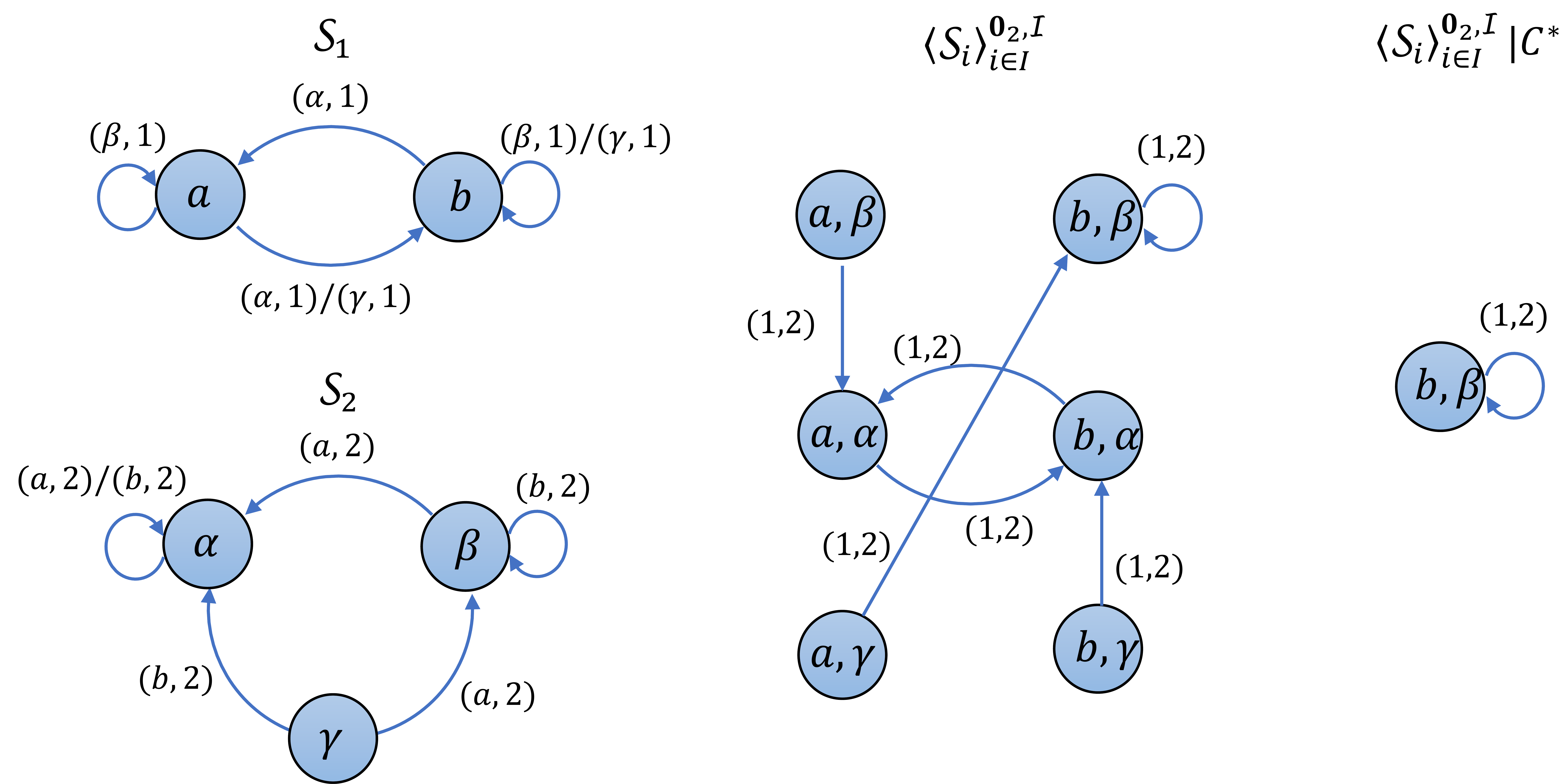}
		\vspace{-0.2em}
		\caption{Two transition systems $S_1$ and $S_2$, the composed system $\langle S_i \rangle_{i\in I}^{\mathbf{0}_2,\mathcal{I}}$ and the controlled system $\langle S_i \rangle_{i\in I}^{\mathbf{0}_2,\mathcal{I}} |C^*$.}\label{fig:example1}
		\vspace{0.2cm}
		
		\includegraphics[scale=0.12]{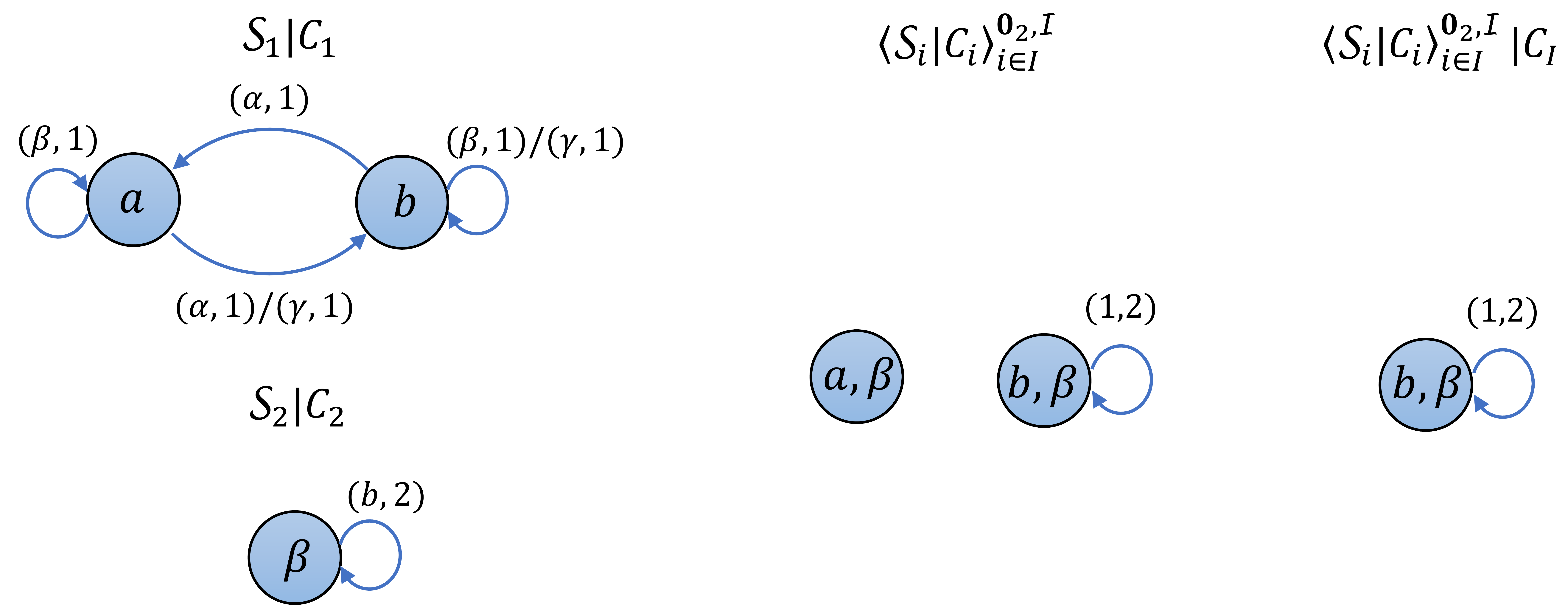}
		\vspace{-0.2em}
		\caption{The controlled components $S_1|C_1$ and $S_2|C_2$, the composition of controlled components $\langle S_i|C_i \rangle_{i\in I}^{\mathbf{0}_2,\mathcal{I}}$ which is finally controlled by $C_I$ ($\langle S_i|C_i \rangle_{i\in I}^{\mathbf{0}_2,\mathcal{I}}|C_I$).}\label{fig:example2}

	\end{center}
	\vspace{-0.3cm}
\end{figure}

\begin{example}
	Consider the transition systems $S_1=(\{a,b\},\{a,b\},\{1\},\{\alpha,\beta,\gamma\},\Delta_1,\{a,b\},Id)$ and $S_2=(\{\alpha,\beta,\gamma\},\{\alpha,\beta,\gamma\},\{2\},\{a,b\},\Delta_2,\{\alpha,\beta,\gamma\},Id)$ as shown in Figure~\ref{fig:example1}, where $Id$ is the identity map. Let the interconnection relation be $\mathcal{I}=\{(1,2),(2,1)\}$. 
	Since $Y_1 \subseteq U^{\Int}_2$ and $Y_2 \subseteq U^{\Int}_1$, the components $S_1$ and $S_2$ are compatible for exact composition with respect to $\mathcal{I}$. Let $\langle S_{i} \rangle_{i\in I}^{\mathbf{0}_2,\mathcal{I}}$ be the composed system.
	Let the global safety specification for the system $S$ define by $\mathfrak S=\mathfrak S_1\times \mathfrak S_2$, where $\mathfrak S_1=\{a,b\}$ and $\mathfrak S_2=\{\beta\}$. The classical safety approach directly synthesize a maximal safety controller for the system $\langle S_{i} \rangle_{i\in I}^{\mathbf{0}_2,\mathcal{I}}$. An illustration of the controlled system $\langle S_{i} \rangle_{i\in I}^{\mathbf{0}_2,\mathcal{I}}|C^*$ is given in Figure~\ref{fig:example1}. In the proposed bottom-up approach, see Figure~\ref{fig:example2}, we first start by synthesizing a local controller $C^*_i$ for the component $S_i$ against the safety specification $\mathfrak S_i$ to obtain $S_i|C^*_i$. We then compose the local controlled components by computing $\langle S_{i}|C_i \rangle_{i\in I}^{\mathbf{0}_2,\mathcal{I}}$. Finally we synthesize a global safety controller $C_I$ for the system $\langle S_{i}|C_i \rangle_{i\in I}^{\mathbf{0}_2,\mathcal{I}}$ against the safety specification $\mathfrak S$. In the classical safety synthesis, we need to compute the safety controller for the system $\langle S_{i} \rangle_{i\in I}^{\mathbf{0}_2,\mathcal{I}}$, which is made of $6$ states and $6$ transitions. In the proposed bottom-up synthesis, we need to apply the global safety synthesis for the reduced composed system  $\langle S_{i}|C_i \rangle_{i\in I}^{\mathbf{0}_2,\mathcal{I}}$ which is made of $2$ states and $1$ transition\footnote{Let us mention that computational complexity to compute the local controllers $C^*_i$ for components $S_i$ is imperceptible with comparison to the safety synthesis on the global reduced composed system $\langle S_{i}|C_i \rangle_{i\in I}^{\mathbf{0}_2,\mathcal{I}}$.}. Hence, one can notice in this toy example the benefits of the proposed synthesis while ensuring completeness with respect to the classical safety synthesis ($C^*=C_I$).
\end{example}

We start by providing the following auxiliary lemma showing how the maximal safety controller $C^*$ for the composed system $S=\langle S_{i} \rangle_{i\in I}^{M,\mathcal{I}}$ is related to the maximal controllers $C^*_i$ synthesized for the components $S_i$, $i\in I$.

\begin{lemma}
	\label{lem:controllers}
	Let $\{S_i\}_{i\in I}$ be a collection of transition systems compatible for $M$-approximate composition, with $M=(\delta_1,\ldots,\delta_N)$. Let $S_M=\langle S_{i} \rangle_{i\in I}^{M,\mathcal{I}}$ be the composed system. Let $\mathfrak S=\mathfrak S_1\times \ldots \times \mathfrak S_N$ be a safety specification for the composed system and let us assume the following:
	\begin{itemize}
		\item $C^*_i$, $i \in I,$ is the maximal safety controller for $S_i$ enforcing the specification $\mathfrak S_i$;
		\item $C^*$ is the maximal controller for $S_M$ enforcing the safety specification $\mathfrak S$.
	\end{itemize}
	If $(u_1^{\Ext},\ldots,u_N^{\Ext}) \in C^*(x_1,\ldots,x_N)$ and  $d_{U^{\Int}_i}(u^{\Int}_i,\prod_{j\in \mathcal{N}(i)}\{H_j(x_j)\}) \leq \delta_i$ for some $i\in I$ and some $u_i^{\Int}\in U_i^{\Int}$, then we have $(u_i^{\Ext},u_i^{\Int}) \in C^*_i(x_i)$.
\end{lemma}
\begin{proof}
		For $i\in I$, let us define the controller $C_i:X_i\rightrightarrows U^{\Ext}_i \times U^{\Int}_i$ as follows: $(u_i^{\Ext},u_i^{\Int}) \in C_i(x_i)$ if and only if there exists $(x_1,\ldots,x_{i-1},x_{i+1},\ldots,x_N)^T\in X_1\times\ldots\times X_{i-1}\times X_{i+1}\times \ldots \times X_N$ and there exists $(u^{\Ext}_1,\ldots,u^{\Ext}_{i-1},u^{\Ext}_{i+1},\ldots,u^{\Ext}_N)^T \in U^{\Ext}_1\times\ldots\times U^{\Ext}_{i-1}\times U^{\Ext}_{i+1}\times \ldots \times U^{\Ext}_N$ such that $(u_1^{\Ext},\ldots,u_N^{\Ext})^T \in C^*(x_1,\ldots,x_N)$ and $\mathbf d_{U^{\Int}_i}(u^{\Int}_i,\prod_{j\in \mathcal{N}(i)}\{H_j(x_j)\}) \leq \delta_i$.
		Let us prove that $C_i$ is a safety controller for $S_i$ and safety specification $\mathfrak S_i$.
		
		Let $x_i \in \dom(C_i)$, then by construction of $C_i$ we have the existence of  $(x_1,\ldots,x_{i-1},x_{i+1},\ldots,x_N)^T \in X_1\times\ldots\times X_{i-1}\times X_{i+1}\times \ldots \times X_N$ such that $(x_1,\ldots,x_N)^T \in \dom(C^*) \subseteq \mathfrak S=\mathfrak S_1\times \ldots \times \mathfrak S_N$. Hence, $x_i \in \mathfrak S_i$ and $\dom(C_i)\subseteq \mathfrak S_i$, then condition (i) of Definition~\ref{safety} is satisfied. Now let $x_i \in \dom(C_i)$ and $(u_i^{\Ext},u_i^{\Int}) \in C_i(x_i)$. We have the existence of $(x_1,\ldots,x_{i-1},x_{i+1},\ldots,x_N)^T \in X_1\times\ldots\times X_{i-1}\times X_{i+1}\times \ldots \times X_N$ and  $(u^{\Ext}_1,\ldots,u^{\Ext}_{i-1},u^{\Ext}_{i+1},\ldots,u^{\Ext}_N)^T \in  U^{\Ext}_1\times\ldots\times U^{\Ext}_{i-1}\times U^{\Ext}_{i+1}\times \ldots \times U^{\Ext}_N$ such that $(u_1^{\Ext},\ldots,u_N^{\Ext})^T \in C^*(x_1,\ldots,x_N)$ and $d_{U^{\Int}_i}(u^{\Int}_i,\prod_{j\in \mathcal{N}(i)}\{H_j(x_j)\}) \leq \delta_i$. Since $C^*$ is the maximal safety controller for the system $S$ and safety specification $\mathfrak S$, we have that for all $(x_1',\ldots,x_N')^T \in \Delta_{C^*}(x_1,\ldots,x_N,u^{\Ext}_1,\ldots,u^{\Ext}_N)$, $(x_1',\ldots, x_N')^T \in \dom(C^*)$. Hence, for $i \in I$, $x_i'\in \dom(C_i)$ for all $x_i'\in \Delta_i(x_i,u_i^{\Ext},u_i^{\Int})$ where $u_i^{\Int} \in U_i^{\Int}$ is such that $\mathbf d_{U^{\Int}_i}(u^{\Int}_i,\prod_{j\in \mathcal{N}(i)}\{H_j(x_j)\}) \leq \delta_i$. Then, $C_i$ is a safety controller for the component $S_i$ and safety specification $\mathfrak S_i$. 
		
		Now let $(u_1^{\Ext},\ldots,u_N^{\Ext})^T \in C^*(x_1,\ldots,x_N)$. Then from construction of $C_i$, for all $i\in I$, and for all $u_i^{\Int}\in U_i^{\Int}$ such that $\mathbf d_{U^{\Int}_i}(u^{\Int}_i,\prod_{j\in \mathcal{N}(i)}\{H_j(x_j)\}) \leq \delta_i$, we get $(u_i^{\Ext},u_i^{Int}) \in C_i(x_i) \subseteq C^*_i(x_i)$, where the last inclusion follows from the maximality of the controller $C_i^*$ for the component $S_i$ and specification $\mathfrak S_i$, which concludes the proof.
\end{proof}	

Next, we provide theorem showing the completeness of the proposed bottom-up controller synthesis procedure with respect to the maximal monolithic safety controller $C^*$.

\begin{theorem}
	\label{thm:completeness}
	Let $\{S_i\}_{i\in I}$ be a collection of transition systems compatible for $M$-approximate composition, with $M=(\delta_1,\ldots,\delta_N)$. Let $S_M=\langle S_{i} \rangle_{i\in I}^{M,\mathcal{I}}$ be the composed system. Let $\mathfrak S=\mathfrak S_1\times \ldots \times \mathfrak S_N$ be a safety specification for the composed system and let us assume the following:
	\begin{itemize}
		\item $C^*_i$, $i \in I,$ is the maximal safety controller for $S_i$ enforcing the specification $\mathfrak S_i$;
		\item $C^*$ is the maximal controller for $S_M$ enforcing the safety specification $\mathfrak S$;
		\item $C_I$ is the maximal controller for $\langle S_i|C_i^* \rangle_{i\in I}^{M,\mathcal{I}}$   enforcing the safety specification $\mathfrak S$.
	\end{itemize}
	Then, for all $x \in X$, $C^*(x)=C_I(x)$. 
\end{theorem}

\begin{proof}
	From Lemma~\ref{lem:controlledsystem}, we have for all $i \in I$, $S_i|C_i^* \preccurlyeq^{0,0}_{\mathcal{AS}} S_i$, then it follows from Theorem~\ref{thm:altsimulation} that $\langle S_i|C_i^* \rangle_{i\in I}^{M,\mathcal{I}} \preccurlyeq^{0,0}_{\mathcal{AS}} S_M=\langle S_i \rangle_{i\in I}^{M,\mathcal{I}}$. Since $C_I$ is the maximal safety controller for $\langle S_i|C_i^* \rangle_{i\in I}^{M,\mathcal{I}}$ and safety specification $\mathfrak S$ and from definition of the alternating simulation relation~\cite{tabuada2009verification}, we have that $C_I$ is a safety controller for the system $S_M$ and specification $\mathfrak S$. Then from maximality of $C^*$, we have that $C_I(x) \subseteq C^*(x)$ for all $x \in X$.
	
	To prove the second inclusion, let us first show that $S_M|C^* \preccurlyeq^{0,0}_{\mathcal{AS}}  \langle S_i|C_i^* \rangle_{i\in I}^{M,\mathcal{I}}$, with $S_M|C^*=(X_{C^*},X_{C^*}^{0},U_{C^*}^{\Ext}$, $U_{C^*}^{\Int},\Delta_{C^*},Y_{C^*},H_{C^*})$ and $\langle S_i|C_i^* \rangle_{i\in I}^{M,\mathcal{I}}=(X_{I},X_{I}^{0},U_{I}^{\Ext},U_{I}^{\Int},\Delta_{I},Y_{I},H_{I})$. Let the relation $\mathcal{R}$ defined by
	$\mathcal{R}=\{(x,\overline{x}) \in X_I \times X_{C^*}\mid x=\overline{x}\}.$
	
	From Lemma~\ref{lem:controllers} we have that $\dom(C^*)\subseteq \dom(C_1)\times \ldots \times \dom(C_N)$, hence, $X_{C^*}^0=X^0\cap\dom(C^*) \subseteq X_{I}^{0}= X^0\cap(\dom(C^*_1)\times \ldots \times \dom(C^*_N))$ and the first condition of Definition~\ref{Def:altsimu} is satisfied. Let $(x,\overline{x}) \in\mathcal{R}$, we have that $x=\overline{x}$, hence, $H_{C^*}(x)=H_I(\overline{x})=H(x)$ and condition (ii) of Definition~\ref{Def:altsimu} is satisfied. Now, let $(x,\overline{x}) \in\mathcal{R}$ and $(u_1^{\Ext},\ldots,u_N^{\Ext})^T \in U^a_{S_M|C^*}(x)$. Then $(u_1^{\Ext},\ldots,u_N^{\Ext})^T \in U^a_{S_M}(x)$ and  $(u_1^{\Ext},\ldots,u_N^{\Ext})^T \in C^*(x)=C^*(x_1,\ldots,x_N)$. We have from Lemma~\ref{lem:controllers} that for all $i \in I$, $(u_i^{\Ext},u_i^{\Int}) \in C^*_i(x_i)$ for any $u_i^{\Int} \in U_i^{\Int}$ satisfying $\mathbf d_{U^{\Int}_i}(u^{\Int}_i,\prod_{j\in \mathcal{N}(i)}\{H_j(x_j)\}) \leq \delta_i$. Then, by construction of the transition systems $S_M|C^*$ and $\langle S_i|C_i^* \rangle_{i\in I}^{M,\mathcal{I}}$, we have for all $\overline{x}'=(\overline{x}_1',\ldots,\overline{x}_N')^T \in \Delta_{I}(\overline{x}_1,\ldots,\overline{x}_N,u_1^{\Ext},\ldots,u_N^{\Ext})$, $(\overline{x}_1',\ldots,\overline{x}_N')^T \in \Delta_{C^*}(\overline{x}_1,\ldots,\overline{x}_N,u_1^{\Ext},\ldots,u_N^{\Ext})=\Delta_{C^*}(x_1,\ldots,x_N,u_1^{\Ext},\ldots,u_N^{\Ext})$. Hence there exists $x' \in \Delta_{C^*}(x_1,\ldots,x_N,u_1^{\Ext},\ldots,u_N^{\Ext})$ such that $x'=\overline{x}'$. Then, $(x',\overline{x}') \in \mathcal{R}$ and condition (iii) of Definition~\ref{Def:altsimu} is satisfied. 
	
	Since $C^*$ is the maximal safety controller for $S_M|C^*$ and safety specification $\mathfrak S$ and from the definition of the alternating simulation relation~\cite{tabuada2009verification}, we have that $C^*$ is a safety controller for the system $\langle S_i|C_i \rangle_{i\in I}^{M,\mathcal{I}}$ and specification $\mathfrak S$. Then from maximality of $C_I$, we have that $C^*(x) \subseteq C_I(x)$ for all $x \in X$. Then for all $x \in X$, $C^*(x)=C_I(x)$.
\end{proof}

In order to clarify the link between the results of Sections~\ref{sec:compo} and~\ref{sec:control}, the following remark is in order.
\begin{remark}
	Given a system $S$ made of $N$ interconnected components $S_i$, $i \in I$, we have the following scenarios:
	\begin{itemize}
		\item  If we have a decomposable safety specification, $\mathfrak S=\mathfrak S_1\times \ldots\times \mathfrak S_N$, we start by constructing a local abstraction $\hat{S}_i$ for each component $S_i$, synthesize a local safety controller $C_i$ for each local abstraction $\hat{S}_i$ and safety specification $\mathfrak S_i$, compose the local controlled abstractions (by computing $\langle \hat{S}_i|C_i \rangle_{i\in I}^{M,\mathcal{I}}$), using the result of Theorems~\ref{thm:altsimulation} and~\ref{thm:simulation}, and then synthesize a global safety controller for $\langle \hat{S}_i|C_i \rangle_{i\in I}^{M,\mathcal{I}}$ against the safety specification $\mathfrak S$, while ensuring in view of Theorem~\ref{thm:completeness} the completeness of the proposed bottom-up synthesis approach with respect to the monolithic abstraction-based synthesis (cf. traffic flow example in Subsection~\ref{examp:traffic}).
		\item If we have any other type of specification $\mathfrak S$, such as reachability, stability, non-decomposable safety specification or more general LTL formula, we start by constructing a local abstraction $\hat{S}_i$ for each component $S_i$, use the result of Theorems~\ref{thm:altsimulation} and~\ref{thm:simulation} to compute a global abstraction $\hat{S}=\langle \hat{S}_i|C_i \rangle_{i\in I}^{M,\mathcal{I}}$ from the local ones, and then synthesize the controller monolithically (cf. DC microgrids example in Subsection~\ref{examp:microgrid}).
	\end{itemize}
\end{remark}

\section{Numerical examples}
\label{sec:examples}

{In this section, we demonstrate the effectiveness of the proposed approach on two control problems: a DC microgrid and a road traffic control problem. The objective of the first example is to illustrate the speed-up that can be attained using the compositional abstraction framework proposed in Section~\ref{sec:compo}. In the second example, we show how the proposed framework can be applied to a more complex example, on which different abstraction techniques are used for different components to show the efficacy of proposed results in Section~\ref{sec:compo}. Additionally, we will also show the benefits of the bottom-up safety synthesis approach proposed in Section~\ref{sec:control}.} In the following, the numerical implementations have been done in MATLAB and a computer with processor 2.7 GHz Intel Core i5, Memory 8 GB 1867 MHz DDR3.

\subsection{DC microgrids}
\label{examp:microgrid}
Direct-current (DC) microgrids have been recognized as a promising choice for the redesign of distribution systems, which are undergoing relevant changes due to the increased penetration of photovoltaic modules, batteries and DC loads. A microgrid is an electrical network gathering a
combination of generation units, loads and energy storage
elements.
Here, we use the DC microgrid model proposed in~\cite{zonetti2019symbolic}.

\subsubsection{Model description and control objective}

We represent a microgrid as a directed graph  $\mathcal{G}(\mathcal{N},\mathcal{E},\mathcal{B})$,  where: $\mathcal{N}$ is the set of nodes, with cardinality $n$; $\mathcal{E}$ is the set of edges, with cardinality $ t$ and $\mathcal{B}\in\mathbb{R}^{n\times t}$ is the incidence matrix capturing the graph topology. The edges correspond to the transmission lines, while the nodes correspond to the buses where the power units are interfaced. 
The weighted interconnection topology is equivalently captured by the Laplacian matrix $\mathcal{L}:=\mathcal{B}G_T\mathcal{B}^\top\in\mathbb{R}^{n\times n}$, with $G_T:=\mathrm{diag}(G_e)\in\mathbb{R}^{t\times t}$, where $G_e$ denotes the conductance associated to the edge $e\in\mathcal{E}$. We further define $\mathcal{N}_{S}$ as the subset of nodes associated to controllable power units (sources), \textit{i.e.} the generation and energy storage units, with cardinality $m$, and $\mathcal{N}_L$ as the subset of nodes associated to non-controllable power units (loads), with cardinality $n-m$. The interconnected dynamics of the voltage buses are:
\begin{equation}\label{eq:compactagent}
C\dot V=-(\mathcal{L}+G)V+\sigma,
\end{equation}
where $V:=\mathrm{col}(v_i)\in{\mathbb{R}_{0}^+}^n$ denotes the collection of (positive) bus voltages, $\sigma:=\mathrm{col}(\sigma_i)\in\mathbb{R}^{n}$ denotes the collection of input currents and $C:=\mathrm{diag}(C_i)\in\mathbb{R}^{n\times n}$, $G:=\mathrm{diag}(G_i)\in\mathbb{R}^{n\times n}$ are matrices denoting the bus capacitances and conductances. Input currents are given by:
\begin{equation}\label{eq:current_injection}
{\sigma}_i=(\mathsf (1-b_i)P_i+b_iu_i)/v_i,\qquad i\in\mathcal{N},
\end{equation}
with: control input  $u_i\in\mathcal{U}_i$, where $\mathcal{U}_i:=[\underline{u}_i,\overline{u}_i]\subset\mathbb{R}^+$; $b_i\in\{0,1\}$, where $b_i=1$, if $i\in\mathcal{N}_S$ and $b_i=0$ otherwise; and $P_i$ is a bounded time-varying demand $P_i \in \mathcal{P}_i=[\underline{P}_i,\overline{P}_i]$. By replacing \eqref{eq:current_injection} into \eqref{eq:compactagent}, the overall system can be rewritten in compact form as
\begin{equation}\label{eq:microgrid}
\dot{V} \in  f(V,u) =  -C^{-1}\left[(\mathcal L +G)V+\begin{bmatrix}
u\\
\mathcal{P}
\end{bmatrix}\oslash V\right]
\end{equation}
with state vector $V\in{\mathbb{R}_{0}^+}^n$; control input $u\in\mathcal U$, where $\mathcal{U}:=\prod_{i }\mathcal{U}_i$; disturbance input  $\mathcal{P}:= \prod_{i}\mathcal{P}_i$; and where $\oslash$ denotes the element-wise division of matrices.

\medskip

The safe set is given by $\mathfrak{S}=[V^\mathrm{nom}-\delta, V^\mathrm{nom}+\delta]^n$ and means that the voltage $V$ of the system need to be kept sufficiently near to the nominal value $V_{nom}>0$ up to a given precision $\delta>0$.

\subsubsection{Abstraction and controller synthesis}
We consider a five-terminal DC microgrid as the one depicted in Figure~\ref{fig:grid}. We assume that two units, namely Units  $2$ and  $3$, are equipped with a primary control layer, while the remaining three units, Units  $1$,  $4$, and $5$ correspond to loads with demand varying steadily around a constant power reference. The latter can be thus interpreted as constant power loads affected by noise. 
\begin{figure}[t!]
	\centering
	\includegraphics[scale=0.6]{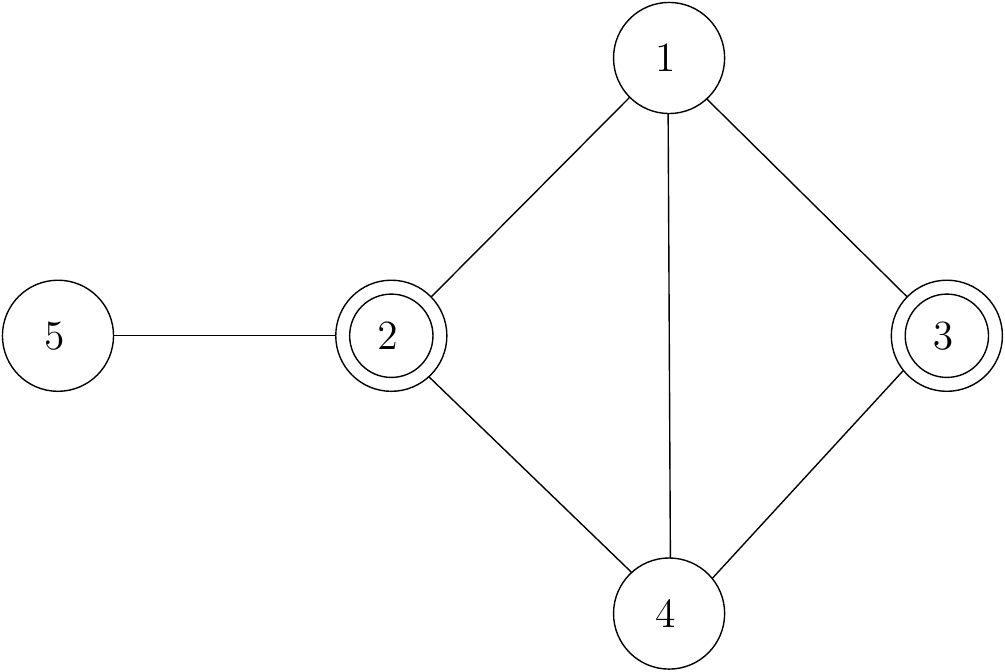}
	\caption{The five units architecture used for the simulations. Circles correspond to loads and sources are denoted by double circles. Solid lines denote the transmission lines.}
	\label{fig:grid}
\end{figure}
The considered bus parameters are $C_1=2.2~\mathrm{\mu F}$, $C_2=1.9~\mathrm{\mu F}$, $C_3=1.5~\mathrm{\mu F}$, $C_4=C_5=1.7~\mathrm{\mu F}$ and the network parameters are $G_{12}=5.2~\Omega^{-1}$, $G_{13}=4.6~\Omega^{-1}$, $G_{14}=4.5~\Omega^{-1}$, $G_{24}=6~\Omega^{-1}$, $G_{25}=3.1~\Omega^{-1}$, $G_{34}=5.6~\Omega^{-1}$, $G_{15}=G_{23}=G_{35}=G_{45}=0~\Omega^{-1}$.
\begin{figure}[t!]
	\centering
	\includegraphics[scale=0.5]{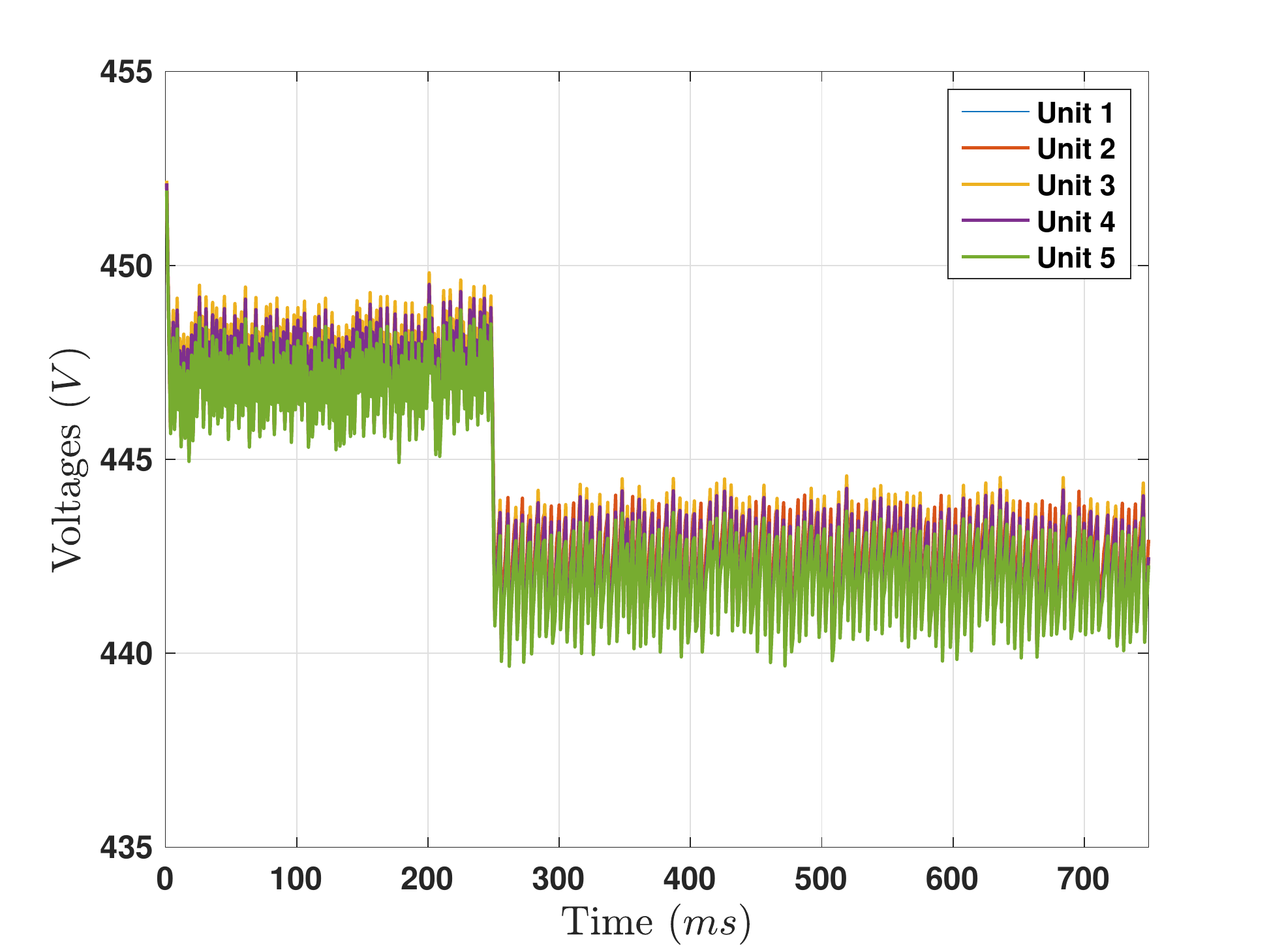} \caption{Voltage responses of the five units.}
	\label{fig:5units}
\end{figure}
The system is supposed to operate within a region with grid nominal voltage $V^\mathrm{nom}=450\; \mathrm V$ and $\delta=0.025 V^\mathrm{nom}$. We use the symbolic approach presented in~\cite{meyer2015adhs}, while exploiting the monotonicity property of the DC grid~\cite{zonetti2019decentralized}, we select sampling period for the abstractions $\tau=0.1$ milliseconds, which corresponds to the clock of the controller to be designed. Discretization parameters are  $n_d=5$ and $n_u=5$ denoting the number of discrete states and inputs, respectively, for each dimension. 

We consider two scenarios. In the first case, we assume that Unit $5$ is disconnected from the grid and the grid is made of $4$ units $I=\{1,2,3,4\}$. We compute local abstraction $\hat{S}_i$ for each Unit $S_i$, $i \in I$, each abstraction $\hat{S}_i$ is related to the original system $S_i$, $i\in I$, by an $(\varepsilon_i,\mu_i)$-approximate alternating simulation relation, with $\varepsilon_i=4.5$ and $\mu_i=0$. We then compose the local abstractions in order to compute the global abstraction using an $\hat{M}$-approximate composition, with $\hat{M}=(4.5, 4.5, 4.5, 4.5)$. Hence, in view of Theorem~\ref{thm:altsimulation}, we have that $\hat{S}\preccurlyeq^{4.5,0}_{\mathcal{AS}}S$ \footnote{{Given the safety specification for the original system $\mathfrak{S}$ and since the original system is related to the compositional abstraction by an $\varepsilon-$approximate alternating simulation relation, the abstract specification is a deflated version of the original one, and is given by $\mathfrak{\hat{S}}=[V^\mathrm{nom}-\delta+\varepsilon, V^\mathrm{nom}+\delta-\varepsilon]^n$}.}, where $S=\langle S^i \rangle_{i\in I}^{\mathbf{0}_4,\mathcal{I}}$ and $\hat{S}=\langle \hat{S}_i \rangle_{i\in I}^{\hat{M},\mathcal{I}}$.

The computation time of the abstractions of the four components $\{1,2,3,4\}$ are given by $5$ seconds, $9$ seconds, $8$ seconds and $4$ seconds, respectively, and the composition of the global abstraction from local ones using an approximate composition takes $15$ seconds. This resulted in $41$ seconds to compute an abstraction compositionally. Constructing an abstraction for the full model
monolithically, using the same discretization parameters, took
$154$ seconds. Hence, the proposed compositional approach was
three times faster in this scenario.

In the second scenario, Unit $5$ is connected to the grid, we use the same numerical parameters as in the first scenario. In this case, the computation time of the abstraction of the five components $\{1,2,3,4,5\}$ are given by $5$ seconds, $43$ seconds, $8$ seconds, $4$ seconds, and $3$ seconds, respectively, and the composition of the global abstraction from local ones using an approximate composition takes $32$ minutes. Let us mention that with comparison to the previous scenario (where only Units 1 to 4 are considered), only the computation time of Unit 2 is modified, since it is the only Unit connected to Unit 5 (see Figure~\ref{fig:grid}). Using the same numerical values, the direct computation of the monolithic abstraction takes $13$ hours, which shows the practical speedups that can be attained using the compositional approach. For the same example, the number of controllable states for the monolithic and compositional schemes is the same and is given by 3125. The size of resulting controllers (number of transitions) for the monolithic and compositional schemes is given by 18252 and 11040, respectively. Hence, it can be seen that the speed-up of the computation of the abstraction using the proposed compositional approach here results in losing the abstraction precision in comparison with the monolithic one, which is observed from the size of the obtained controllers.

We then synthesize a safety controller for the computed abstraction. The synthesis of the symbolic controller takes $30$ seconds. To validate our controller, we assume that the load power demands for Unit 1, Unit 4 and Unit 5 are as follows. Unit 1 is demanding $0.3\;\mathrm{kW}$ from $0$ to $250$ milliseconds, immediately after stepping up to $1\;\mathrm{kW}$. Unit 4 on the other hand is supposed to be characterized by a demand of $0.3\;\mathrm{kW}$ from $0$ to $250$ milliseconds, then a constant demand of $1\; \mathrm{kW}$ from $250$ milliseconds to $750$ milliseconds. Finally, Unit 5 is characterized by a demand of $0.4\;\mathrm{kW}$ from $0$ to $250$ milliseconds, then a constant demand of $1\; \mathrm{kW}$ from $250$ milliseconds to $750$ milliseconds. All demands are affected by small noise. Source power injections are positive and both limited at $8\; \mathrm {kW}$. The controller is implemented via a microprocessor of clock period $\tau=0.1$ milliseconds. Voltage responses for different units are illustrated in Figure \ref{fig:5units}. As expected, the controller guarantees that voltages are kept sufficiently near the nominal value. 
\subsection{ Traffic flow model}
\label{examp:traffic}
For the second example, we consider a popular macroscopic
traffic flow model used for the analysis of highway traffic behavior is the cell transmission model (CTM) introduced in \cite{daganzo1994cell}.
\subsubsection{Model description and control objective}
Here we consider the traffic flow model as shown schematically in Figure \ref{fig:traffic} and described by 
\begin{align}
x_1(k\hspace{-0.2em}+\hspace{-0.2em}1)&\hspace{-0.2em}=\hspace{-0.2em}(1\hspace{-0.2em}-\hspace{-0.2em}\frac{Tv}{1.6l})x_1(k)\hspace{-0.2em}+\hspace{-0.2em}5u_1(k),\nonumber\\
x_2(k\hspace{-0.2em}+\hspace{-0.2em}1)&\hspace{-0.2em}=\hspace{-0.2em}\frac{Tv}{l}x_1(k)\hspace{-0.2em}+\hspace{-0.2em}(1\hspace{-0.2em}-\hspace{-0.2em}\frac{Tv}{l}\hspace{-0.2em}-\hspace{-0.2em}q)x_2(k)\hspace{-0.2em}+\hspace{-0.2em}\frac{Tv}{l}x_4(k),\nonumber\\
x_3(k\hspace{-0.2em}+\hspace{-0.2em}1)&\hspace{-0.2em}=\hspace{-0.2em}\frac{Tv}{l}x_2(k)\hspace{-0.2em}+\hspace{-0.2em}(1\hspace{-0.2em}-\hspace{-0.2em}\frac{Tv}{l}\hspace{-0.2em}-\hspace{-0.2em}q)x_3(k)\hspace{-0.2em}+\hspace{-0.2em}8u_2(k),\nonumber\\
x_4(k\hspace{-0.2em}+\hspace{-0.2em}1)&\hspace{-0.2em}=\hspace{-0.2em}\frac{Tv}{l}x_3(k)\hspace{-0.2em}+\hspace{-0.2em}(1\hspace{-0.2em}+\hspace{-0.2em}\frac{Tv}{l}\hspace{-0.2em}-\hspace{-0.2em}q)x_4(k)\hspace{-0.2em}+\hspace{-0.2em}8u_3(k),
\end{align}
\begin{figure}[t]
	\centering
	\includegraphics[scale=0.11]{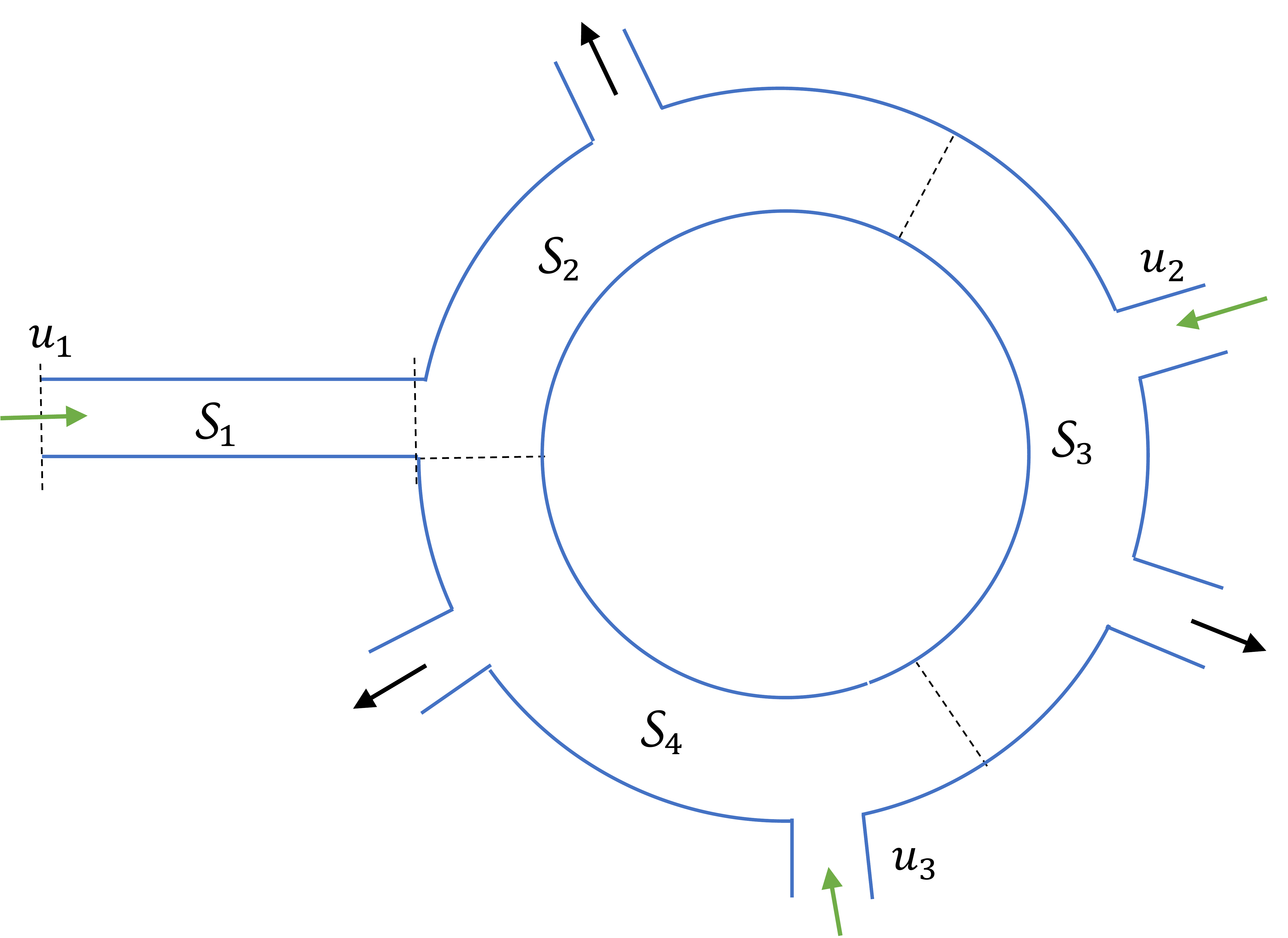} \caption{Traffic flow model schematic.}
	\label{fig:traffic}
	\vspace{-0.3cm}
\end{figure}
where the state $x_i$, $i\in\{1,2,3,4\}$, represents the density of traffic in $i^{th}$ section of road given in vehicles per section, $l=0.25\ km$ is the length of road, $v=70\ km/hr$ is the flow speed, $T=\frac{10}{3600}$ hours is the discrete time interval, and $q=0.25$ is the ratio representing the percentage of vehicles leaving the section of road. The control inputs $u_1,u_2,u_3\in U=\{0,1\}$, where 0 represents red signal and 1 represents green signal in the traffic model. We consider the compact state-space $X=[0,30]^4$. The control objective is to synthesize controller to keep states in a safe region given by $\mathfrak{S}=[2,25]\times[5,25]^3$.
\subsubsection{Abstraction and controller synthesis}
We consider four subsystems $S_i$, $i \in I=\{1,2,3,4\}$, corresponding to four sections in traffic network model. To demonstrate the effectiveness of the proposed result on bottom-up safety controller synthesis, we compare results obtained using monolithic safety synthesis and bottom-up safety synthesis on compositional abstractions. 
For constructing abstractions, we construct an $\varepsilon$-approximate bisimilar abstraction of $S_1$ called $\hat{S}_1$ using state-space discretization-free abstraction techniques as discussed in \cite{le2013mode,zamani2017towards} using tool QUEST \cite{jagtap2017quest}. For construction of $\hat S_1$, with precision $\varepsilon_1=0.0016$ and $\mu_1=0$, we consider $\mathsf{U}=\{0,1\}$, length of input sequence $N=8$, and source state $x_s=10$ (for description and computation of parameters $\varepsilon_1$ and $N$, see \cite{le2013mode} and \cite{zamani2017towards}). The abstractions $\hat S_i$, $i\in\{2,3,4\}$, are computed by utilizing partitions of the state-space as shown in~\cite{meyer2015adhs,reissig2017feedback}, each abstraction $\hat{S}_i$ is related to the original component $S_i$ by an $(\varepsilon_i,\mu_i)$-approximate alternating simulation relation, with $\varepsilon_i=0.1$ and $\mu_i=0$. Note that since the subsystem $S_1$ is incrementally input-state stable and do not have any internal input, the input set of the component $S_1$ is much smaller compared to other components. In such a scenario, state-space discretization-free abstractions are efficient compared to state-space discretization based abstractions (see Section 4.D in~\cite{le2013mode} and Section 5.4 in \cite{zamani2017towards} for detailed discussion). 

Monolithic and proposed bottom-up approaches to safety synthesis are then compared. In the first one, we compute the global compositional abstraction $\hat{S}=\langle \hat{S}_i \rangle_{i\in I}^{\hat{M},\mathcal{I}}$ by composing local abstractions with a composition parameter $\hat M=(0.0016,0.1,0.1,0.1)$. We then monolithically synthesize a safety controller for the global abstraction $\hat{S}$ with the safe set { $\hat{\mathfrak{S}}=[2+\varepsilon,25-\varepsilon]\times[5+\varepsilon,25-\varepsilon]^3$, with $\varepsilon=\max_{i\in \{1,2,3,4\}}\varepsilon_i=0.1$,} using maximal fixed point algorithm \cite{tabuada2009verification}. The total computation time required for obtaining the monolithic safety controller is $9$ hours and $20$ minutes.

In the second approach, we start from the local abstractions $\hat{S}_i$, $i \in I$, and first compute safety controllers $C^*_i$ for each abstraction $\hat S_i$, $i\in I$, with local safety specification $\hat{\mathfrak{S}}_i$, {where $\hat{\mathfrak{S}}_1=[2+\varepsilon_1, 25-\varepsilon_1]=[2.0016,24.9984]$ and $\hat{\mathfrak{S}}_j=[2+\varepsilon_j, 25-\varepsilon_j]=[5.1,24.9]$ for $j\in \{2,3,4\}$}. Then we compose the local controlled components $S_i|C_i^*$, $i\in I$, with an $\hat M$-approximate composition with $\hat M=(0.0016,0.1,0.1,0.1)$ given as $\langle S_i|C_i^* \rangle_{i\in I}^{\hat M,\mathcal{I}}$. Then, as a final step we synthesize a safety controller for $\langle S_i|C_i^* \rangle_{i\in I}^{\hat M,\mathcal{I}}$ against the safety specification $\hat{\mathfrak{S}}$. The total computation time required for obtaining the safety controller using bottom-up safety synthesis is $2$ hours and $25$ minutes which is almost $4$ times faster than the monolithic synthesis case. The Figure \ref{fig:resp_traffic} shows the evolution of traffic densities in each section of the road starting from the initial condition $x=[14,15,20,16]$ using controller obtained by proposed bottom-up approach. One can readily see that all the trajectories evolve within the safe region.   
\begin{figure}[t!]
	\centering
	\includegraphics[scale=0.40]{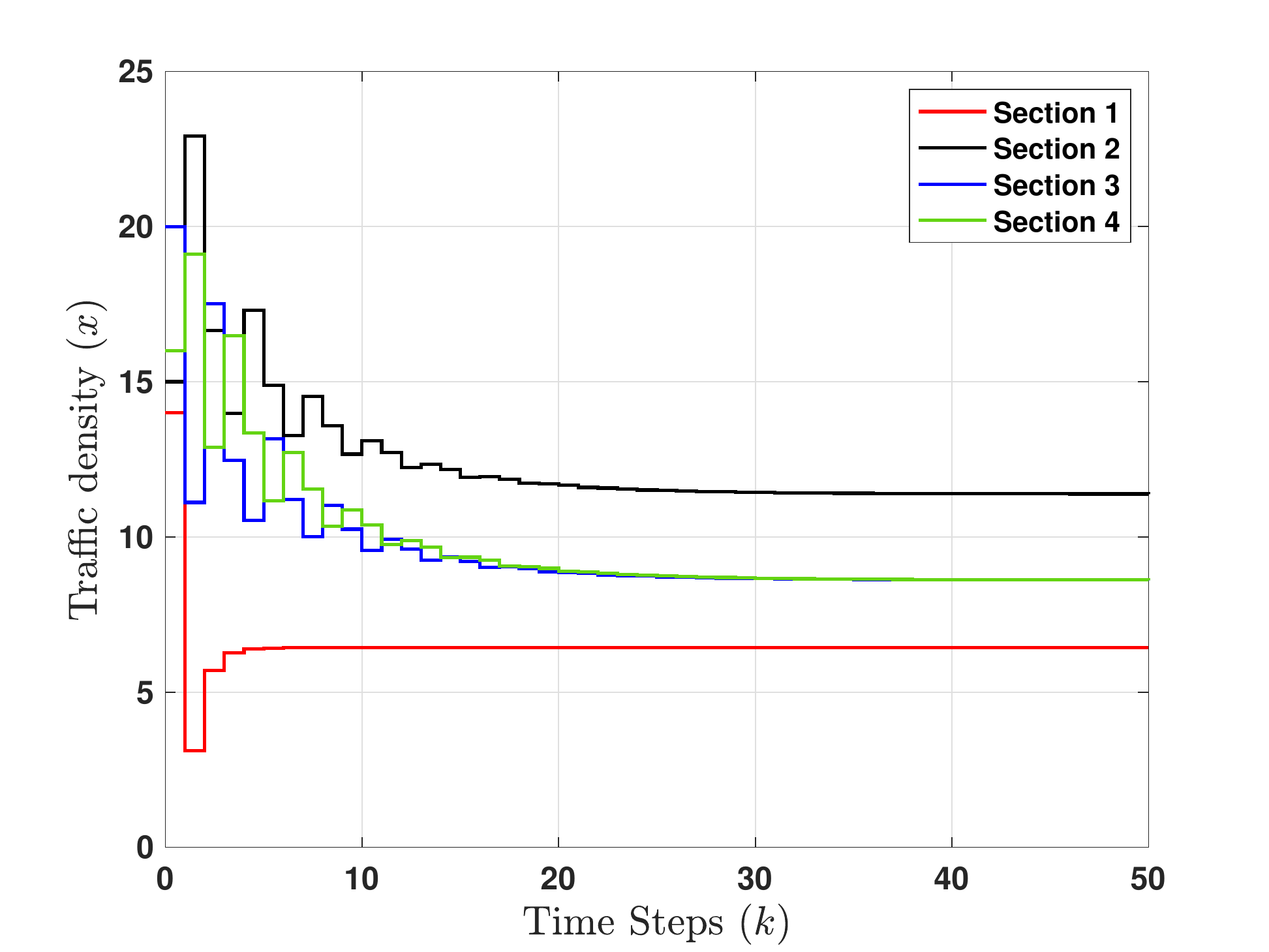} \caption{The evolution of traffic densities in each section of the road.}
	\label{fig:resp_traffic}
	\vspace{-0.3cm}
\end{figure}

\section{Conclusion}
In this paper, we proposed a compositional abstraction-based synthesis approach for interconnected systems. We introduce a notion of approximate composition that allows composing different types of abstractions. Moreover, we provided compositional results based on approximate (alternating) simulation relation and showed how these results can be used for bottom-up safety controller synthesis. Two case studies are given to show the effectiveness of our approach. In future work, we plan to extend the bottom-up synthesis approach from safety to other types of specifications, such as reachability, stability, or more general properties described by temporal logic formulae. Another direction is to go from deterministic relationships to probabilistic ones~\cite{abate2013approximation}, which are more suitable to use when dealing with stochastic systems.

\bibliographystyle{alpha}
\bibliography{references}

\newcommand{\etalchar}[1]{$^{#1}$}
\begin{thebibliography}{PPDBT10}

\bibitem[Aba13]{abate2013approximation}
A.~Abate.
\newblock Approximation metrics based on probabilistic bisimulations for
  general state-space markov processes: {A} survey.
\newblock {\em Electronic Notes in Theoretical Computer Science}, 297:3--25,
  2013.

\bibitem[BJP{\etalchar{+}}12]{bloem2012synthesis}
R.~Bloem, B.~Jobstmann, N.~Piterman, A.~Pnueli, and Y.~Saar.
\newblock Synthesis of reactive (1) designs.
\newblock {\em Journal of Computer and System Sciences}, 78(3):911--938, 2012.

\bibitem[BPDB18]{borri2018design}
A.~Borri, G.~Pola, and M.~D. Di~Benedetto.
\newblock Design of symbolic controllers for networked control systems.
\newblock {\em IEEE Transactions on Automatic Control}, 64(3):1034--1046, 2018.

\bibitem[BYG17]{belta2017formal}
C.~Belta, B.~Yordanov, and E.~A. Gol.
\newblock {\em Formal Methods for Discrete-Time Dynamical Systems}.
\newblock Springer International Publishing, 2017.

\bibitem[CA15]{coogan2015efficient}
S.~Coogan and M.~Arcak.
\newblock Efficient finite abstraction of mixed monotone systems.
\newblock In {\em Proceedings of the 18th International Conference on Hybrid
  Systems: Computation and Control}, pages 58--67. ACM, 2015.

\bibitem[CL09]{cassandras2009introduction}
C.~G. Cassandras and S.~Lafortune.
\newblock {\em Introduction to discrete event systems}.
\newblock Springer Science \& Business Media, 2009.

\bibitem[Dag94]{daganzo1994cell}
Carlos~F Daganzo.
\newblock The cell transmission model: A dynamic representation of highway
  traffic consistent with the hydrodynamic theory.
\newblock {\em Transportation Research Part B: Methodological}, 28(4):269--287,
  1994.

\bibitem[DT15]{dallal2015compositional}
E.~Dallal and P.~Tabuada.
\newblock On compositional symbolic controller synthesis inspired by small-gain
  theorems.
\newblock In {\em 2015 54th IEEE Conference on Decision and Control (CDC)},
  pages 6133--6138, 2015.

\bibitem[Fre05]{frehse2005compositional}
G.~F. Frehse.
\newblock {\em Compositional verification of hybrid systems using simulation
  relations}.
\newblock [Sl: sn], 2005.

\bibitem[GP07]{girard2007}
A.~Girard and G.~J. Pappas.
\newblock Approximation metrics for discrete and continuous systems.
\newblock {\em IEEE Transactions on Automatic Control}, 52(5):782--798, 2007.

\bibitem[GP09]{girard2009hierarchical}
A.~Girard and G.~J. Pappas.
\newblock Hierarchical control system design using approximate simulation.
\newblock {\em Automatica}, 45(2):566--571, 2009.

\bibitem[GPT09]{girard2009approximately}
A.~Girard, G.~Pola, and P.~Tabuada.
\newblock Approximately bisimilar symbolic models for incrementally stable
  switched systems.
\newblock {\em IEEE Transactions on Automatic Control}, 55(1):116--126, 2009.

\bibitem[GSB20]{ghasemi2020compositional}
Kasra Ghasemi, Sadra Sadraddini, and Calin Belta.
\newblock Compositional synthesis via a convex parameterization of
  assume-guarantee contracts.
\newblock In {\em Proceedings of the 23rd International Conference on Hybrid
  Systems: Computation and Control}, pages 1--10, 2020.

\bibitem[HAT17]{hussien2017abstracting}
O.~Hussien, A.~Ames, and P.~Tabuada.
\newblock Abstracting partially feedback linearizable systems compositionally.
\newblock {\em IEEE Control Systems Letters}, 1(2):227--232, 2017.

\bibitem[JDDBP09]{julius2009approximate}
A.~A. Julius, A.~D'Innocenzo, M.~D. Di~Benedetto, and G.~J. Pappas.
\newblock Approximate equivalence and synchronization of metric transition
  systems.
\newblock {\em Systems \& Control Letters}, 58(2):94--101, 2009.

\bibitem[JZ17]{jagtap2017quest}
P.~Jagtap and M.~Zamani.
\newblock {QUEST}: A tool for state-space quantization-free synthesis of
  symbolic controllers.
\newblock In {\em International Conference on Quantitative Evaluation of
  Systems}, pages 309--313. Springer, 2017.

\bibitem[JZ20]{jagtap2020symbolic}
P.~Jagtap and M.~Zamani.
\newblock Symbolic models for retarded jump--diffusion systems.
\newblock {\em Automatica}, 111:108666, 2020.

\bibitem[KAS17]{kim2017small}
E.~S. Kim, M.~Arcak, and S.~A. Seshia.
\newblock A small gain theorem for parametric assume-guarantee contracts.
\newblock In {\em Proceedings of the 20th International Conference on Hybrid
  Systems: Computation and Control}, HSCC '17, pages 207--216, New York, NY,
  USA, 2017. ACM.

\bibitem[KAZ18]{kim2018constructing}
E.~S. Kim, M.~Arcak, and M.~Zamani.
\newblock Constructing control system abstractions from modular components.
\newblock In {\em Proceedings of the 21st International Conference on Hybrid
  Systems: Computation and Control (part of CPS Week)}, pages 137--146. ACM,
  2018.

\bibitem[KVdS10]{kerber2010compositional}
F.~Kerber and A.~Van~der Schaft.
\newblock Compositional analysis for linear systems.
\newblock {\em Systems \& Control Letters}, 59(10):645--653, 2010.

\bibitem[LCGG13]{le2013mode}
E.~Le~Corronc, A.~Girard, and G.~Goessler.
\newblock Mode sequences as symbolic states in abstractions of incrementally
  stable switched systems.
\newblock In {\em 52nd IEEE Conference on Decision and Control}, pages
  3225--3230. IEEE, 2013.

\bibitem[LFM{\etalchar{+}}16]{le2016distributed}
A.~{Le Co{\"e}nt}, L.~Fribourg, N.~Markey, F.~De~Vuyst, and L.~Chamoin.
\newblock Distributed synthesis of state-dependent switching control.
\newblock In {\em International Workshop on Reachability Problems}, pages
  119--133, 2016.

\bibitem[LP17]{lal2017safety}
Ratan Lal and Pavithra Prabhakar.
\newblock Safety analysis using compositional bounded error approximations of
  communicating hybrid systems.
\newblock In {\em 2017 IEEE 56th Annual Conference on Decision and Control
  (CDC)}, pages 2378--2383. IEEE, 2017.

\bibitem[LP19]{lal2019compositional}
Ratan Lal and Pavithra Prabhakar.
\newblock Compositional construction of bounded error over-approximations of
  acyclic interconnected continuous dynamical systems.
\newblock In {\em Proceedings of the 17th ACM-IEEE International Conference on
  Formal Methods and Models for System Design}, pages 1--5, 2019.

\bibitem[MD18]{meyer2018compositional}
P-J. Meyer and D.~V. Dimarogonas.
\newblock Compositional abstraction refinement for control synthesis.
\newblock {\em Nonlinear Analysis: Hybrid Systems}, 27:437--451, 2018.

\bibitem[MGW15]{meyer2015adhs}
P-J. Meyer, A.~Girard, and E.~Witrant.
\newblock Safety control with performance guarantees of cooperative systems
  using compositional abstractions.
\newblock {\em IFAC-PapersOnLine}, 48(27):317--322, 2015.

\bibitem[MGW17]{meyer2017}
P-J. Meyer, A.~Girard, and E.~Witrant.
\newblock Compositional abstraction and safety synthesis using overlapping
  symbolic models.
\newblock {\em IEEE Transactions on Automatic Control}, 63(6):1835--1841, 2017.

\bibitem[MSSM18]{mallik2018compositional}
K.~Mallik, A-K. Schmuck, S.~Soudjani, and R.~Majumdar.
\newblock Compositional synthesis of finite state abstractions.
\newblock {\em IEEE Transactions on Automatic Control}, 2018.

\bibitem[PGT08]{pola2008approximately}
G.~Pola, A.~Girard, and P.~Tabuada.
\newblock Approximately bisimilar symbolic models for nonlinear control
  systems.
\newblock {\em Automatica}, 44(10):2508--2516, 2008.

\bibitem[PPB18]{8115304}
G.~{Pola}, P.~{Pepe}, and M.~D.~D. {Benedetto}.
\newblock Decentralized supervisory control of networks of nonlinear control
  systems.
\newblock {\em IEEE Transactions on Automatic Control}, 63(9):2803--2817, 2018.

\bibitem[PPD16]{pola7}
G.~Pola, P.~Pepe, and {M. D.} {Di Benedetto}.
\newblock Symbolic models for networks of control systems.
\newblock {\em IEEE Transactions on Automatic Control}, 61(11):3663--3668,
  2016.

\bibitem[PPDBT10]{pola2010symbolic}
G.~Pola, P.~Pepe, M.~D. Di~Benedetto, and P.~Tabuada.
\newblock Symbolic models for nonlinear time-delay systems using approximate
  bisimulations.
\newblock {\em Systems \& Control Letters}, 59(6):365--373, 2010.

\bibitem[Rei11]{5770194}
G.~Reissig.
\newblock Computing abstractions of nonlinear systems.
\newblock {\em IEEE Transactions on Automatic Control}, 56(11):2583--2598,
  2011.

\bibitem[RWR17]{reissig2017feedback}
G.~Reissig, A.~Weber, and M.~Rungger.
\newblock Feedback refinement relations for the synthesis of symbolic
  controllers.
\newblock {\em IEEE Transactions on Automatic Control}, 62(4):1781--1796, 2017.

\bibitem[RZ18]{rungger2018compositional}
M.~Rungger and M.~Zamani.
\newblock Compositional construction of approximate abstractions of
  interconnected control systems.
\newblock {\em IEEE Transactions on Control of Network Systems}, 5(1):116--127,
  2018.

\bibitem[SGF18a]{saoud2018contract}
A.~Saoud, A.~Girard, and L.~Fribourg.
\newblock Contract based design of symbolic controllers for interconnected
  multiperiodic sampled-data systems.
\newblock In {\em 2018 IEEE Conference on Decision and Control (CDC)}, pages
  773--779. IEEE, 2018.

\bibitem[SGF18b]{saoud2018contractv}
A.~Saoud, A.~Girard, and L.~Fribourg.
\newblock Contract based design of symbolic controllers for vehicle platooning.
\newblock In {\em Proceedings of the 21st International Conference on Hybrid
  Systems: Computation and Control (part of CPS Week)}, pages 277--278. ACM,
  2018.

\bibitem[SGF18c]{8550622}
A.~{Saoud}, A.~{Girard}, and L.~{Fribourg}.
\newblock On the composition of discrete and continuous-time assume-guarantee
  contracts for invariance.
\newblock In {\em 2018 European Control Conference (ECC)}, pages 435--440,
  2018.

\bibitem[SGF20]{saoud2019contract}
Adnane Saoud, Antoine Girard, and Laurent Fribourg.
\newblock Contract-based design of symbolic controllers for safety in
  distributed multiperiodic sampled-data systems.
\newblock {\em IEEE Transactions on Automatic Control}, 2020.

\bibitem[SJZG18]{SAOUD201813}
A.~Saoud, P.~Jagtap, M.~Zamani, and A.~Girard.
\newblock Compositional abstraction-based synthesis for cascade discrete-time
  control systems.
\newblock {\em 6th IFAC Conference on Analysis and Design of Hybrid Systems
  ADHS}, 51(16):13 -- 18, 2018.

\bibitem[SJZG20]{saoud2020compositional}
Adnane Saoud, Pushpak Jagtap, Majid Zamani, and Antoine Girard.
\newblock Compositional abstraction-based synthesis for interconnected systems:
  An approximate composition approach.
\newblock {\em arXiv preprint arXiv:2002.02014}, 2020.

\bibitem[SNO16]{smith2016interdependence}
Stanley~W Smith, Petter Nilsson, and Necmiye Ozay.
\newblock Interdependence quantification for compositional control synthesis
  with an application in vehicle safety systems.
\newblock In {\em 2016 IEEE 55th Conference on Decision and Control (CDC)},
  pages 5700--5707. IEEE, 2016.

\bibitem[SZ19]{swikir2018compositional}
A.~Swikir and M.~Zamani.
\newblock Compositional synthesis of finite abstractions for networks of
  systems: A small-gain approach.
\newblock {\em Automatica}, 107:551--561, 2019.

\bibitem[Tab09]{tabuada2009verification}
P.~Tabuada.
\newblock {\em Verification and control of hybrid systems: A symbolic
  approach}.
\newblock Springer Science \& Business Media, 2009.

\bibitem[TI08]{tazaki2008bisimilar}
Y.~Tazaki and J.~Imura.
\newblock Bisimilar finite abstractions of interconnected systems.
\newblock {\em Hybrid Systems: Computation and Control}, pages 514--527, 2008.

\bibitem[TPL04]{tabuada2004compositional}
P.~Tabuada, G.~J. Pappas, and P.~Lima.
\newblock Compositional abstractions of hybrid control systems.
\newblock {\em Discrete event dynamic systems}, 14(2):203--238, 2004.

\bibitem[ZA17]{zamani2017compositional}
M.~Zamani and M.~Arcak.
\newblock Compositional abstraction for networks of control systems: A
  dissipativity approach.
\newblock {\em IEEE Transactions on Control of Network Systems},
  5(3):1003--1015, 2017.

\bibitem[ZPMT12]{zamani2012symbolic}
M.~Zamani, G.~Pola, M.~Mazo, and P.~Tabuada.
\newblock Symbolic models for nonlinear control systems without stability
  assumptions.
\newblock {\em IEEE Transactions on Automatic Control}, 57(7):1804--1809, 2012.

\bibitem[ZSGF19a]{zonetti2019decentralized}
D.~Zonetti, A.~Saoud, A.~Girard, and L.~Fribourg.
\newblock Decentralized monotonicity-based voltage control of dc microgrids
  with zip loads.
\newblock {\em IFAC-PapersOnLine}, 52(20):139--144, 2019.

\bibitem[ZSGF19b]{zonetti2019symbolic}
D.~Zonetti, A.~Saoud, A.~Girard, and L.~Fribourg.
\newblock A symbolic approach to voltage stability and power sharing in
  time-varying dc microgrids.
\newblock In {\em European Control Conference (ECC)}, 2019.

\bibitem[ZTA17]{zamani2017towards}
M.~Zamani, I.~Tkachev, and A.~Abate.
\newblock Towards scalable synthesis of stochastic control systems.
\newblock {\em Discrete Event Dynamic Systems}, 27(2):341--369, 2017.

\end{thebibliography}

\end{document}